\documentclass[11pt]{article}
\usepackage[margin=1.1in]{geometry}
\usepackage{amsmath,amssymb,amsthm}
\usepackage{tikz}
\usepackage{graphicx}
\usepackage{url}
\usepackage[hidelinks]{hyperref}
\hypersetup{
  pdftitle={A minimum witness for the 3/2 configuration-linear-program gap in two-weight graph balancing, unique at its size},
  pdfauthor={Adam Y. Shavit},
  pdfsubject={Configuration-linear-program integrality gaps; graph balancing; restricted assignment},
  pdfkeywords={configuration linear program, integrality gap, graph balancing,
               restricted assignment, minimum witness, exhaustive enumeration}
}

\theoremstyle{plain}
\newtheorem{theorem}{Theorem}
\newtheorem{lemma}[theorem]{Lemma}
\newtheorem{definition}[theorem]{Definition}
\newtheorem{proposition}[theorem]{Proposition}
\newtheorem{corollary}[theorem]{Corollary}
\newtheorem{measurement}{Measurement}

\newcommand{\OPT}{\mathrm{OPT}}
\newcommand{\OPTLP}{\mathrm{OPT}_{\mathrm{LP}}}
\newcommand{\Istar}{I^{*}}
\newcommand{\gap}{\mathrm{gap}}

\title{A minimum witness for the $3/2$ configuration-linear-program\\
gap in two-weight graph balancing, unique at its size}

\author{Adam Y. Shavit\\
\small Hunter College and the Graduate Center, CUNY\\
\small \texttt{as1127@hunter.cuny.edu} \quad ORCID 0009-0008-1235-0995}

\date{\today}

\begin{document}
\maketitle

\begin{abstract}
Restricted assignment minimizes makespan, the busiest machine's load, when each
job has one size and a set of allowed machines. Its tightest studied relaxation
is the configuration linear program, which fractionally gives each machine whole
job sets that fit. Its integrality gap, the worst-case ratio of integral optimum
to relaxation value, is unknown in general. On two-weight graph balancing (sizes
1 and 2, each job allowed on at most two machines) the value is known: Jansen,
Land, and Maack (2016) give both bounds. Their Table~1 instance attains $3/2$,
and their Corollary~11 bound of $2 - s/b$ for sizes $s < b$ matches it at
$\{1,2\}$. We ask how small such an instance, a \emph{witness}, can be.

We give a six-job witness $\Istar$ on the complete graph on four machines, each
job an edge between its machines. Unit jobs form a Hamiltonian cycle through all
four; weight-2 jobs the two remaining edges, the complementary perfect
matching. $\Istar$ has integral optimum 3 against relaxation value 2, with one job
fewer than the smallest witness previously in print for this class. We
prove $\Istar$ minimum and unique at its size, on any number of machines. No
instance of the class with at most five jobs reaches gap $3/2$. At six jobs
$\Istar$ is the only witness, up to relabeling machines and adding machines no
job can use.

At seven jobs uniqueness fails: we classify all thirteen witnesses, including
the Jansen--Land--Maack instance. Eight jobs give 154, nine 1,662. Three
machines never suffice: four are necessary for the gap. Others have published comparable
results for the same relaxation in one-dimensional cutting stock, enumerating and
classifying the extremal non-round-up instances for small demand; the
Discussion relates them.

Recognizing witnesses at relaxation value 2, shared by all ours, is
coNP-complete, so no min--max characterization with polynomially checkable
conditions exists unless $\mathrm{NP} = \mathrm{coNP}$. Every feasibility
decision behind the exhaustive claims was made twice, in floating-point and
exact rational arithmetic, agreeing fully. We accept the pipeline's
negatives only after it rediscovers $\Istar$.
\end{abstract}

\section{Introduction}

Lenstra, Shmoys, and Tardos \cite{LST90} gave the factor-2 algorithm for
makespan minimization on unrelated machines, together with a
$3/2$ NP-hardness bound. Neither endpoint of $[3/2, 2]$ has essentially moved in thirty-six
years. The sharpest relaxation available is the \emph{configuration linear
program}, which places fractional weight on per-machine
\emph{configurations}---sets of jobs fitting within a target makespan---
rather than on individual job--machine pairs. Its integrality gap, the
worst-case ratio of the true optimum to the relaxation's bound, limits every
algorithm that rounds it: no such algorithm can guarantee a ratio below
the gap.

For restricted assignment the gap is known only to lie in $[3/2, 11/6]$
---the lower bound realized by \cite[Table~1]{JLM16}, the upper bound
\cite[Theorem~1.1]{JR17}, sharpening \cite{Svensson11}'s $33/17$---and
closing it is a named open problem. On some subclasses the value is already settled, and there a different question opens: how small an instance attaining it can be. We ask it for two-weight graph balancing, where the value has
been $3/2$ since 2016 (Section~\ref{sec:known}).

\paragraph{Contributions.}
\begin{enumerate}
\item \textbf{A smaller witness.} For two-weight graph balancing---every job
allowed on at most two machines, sizes in $\{1,2\}$---we exhibit $\Istar$,
an explicit six-job instance with gap exactly $3/2$
(Theorem~\ref{thm:istar}). The smallest witness \emph{of this class}
previously in print, Table~1 of \cite{JLM16}, has seven jobs; the
Discussion names a five-job near miss that lies outside the class, in
unrelated graph balancing.
\item \textbf{Minimality and uniqueness.} No instance of the class with at
most five jobs attains gap $3/2$, on any number of machines
(Theorem~\ref{thm:min}); at six jobs, again on any number of machines,
$\Istar$ is the unique witness up to machine relabeling and machines no
job can use (Theorem~\ref{thm:unique}). At seven jobs we classify all
witnesses: exactly thirteen (Theorem~\ref{thm:seven}), among them the
instance of \cite{JLM16} and eight connected witnesses not found in the
literature survey described in Section~\ref{sec:disc}; at eight jobs we
count them: exactly 154
(Theorem~\ref{thm:eight}), and at nine exactly 1{,}662
(Theorem~\ref{thm:nine}), of which 978 are connected. Three machines never suffice, at any size
(Theorem~\ref{prop:mach}), so $\Istar$ is minimum in machines as well as in
jobs. Our literature survey found no earlier minimality, uniqueness, or
classification result for gap instances of the configuration linear program
of a scheduling problem. For the same
relaxation in one-dimensional cutting stock, \cite{KRSK15} prove by
exhaustive enumeration over equivalence classes that every instance of
demand at most nine has the integer round-up property, exhibit classes
failing it at ten, and classify the extremal instances up to demand
eleven. Section~\ref{sec:disc} sets out what separates the two questions.
This, not the gap value, is
what the note contributes: the value for this class is already determined
by \cite{JLM16} (Section~\ref{sec:known}).
\item \textbf{A verification standard.} Every feasibility decision behind
the exhaustive claims was made twice, once in floating point and once in
exact rational arithmetic, and the enumeration pipeline was validated end
to end by requiring it to rediscover $\Istar$.
\end{enumerate}

\section{Preliminaries}

An instance is a nonempty finite job set $J$ with sizes $p_j$ and nonempty
allowed-machine sets $M(j)$. \textbf{All sizes are positive integers throughout}; consequently
every configuration load is an integer, and the least feasible threshold is
attained at an integer. $\OPT(I)$ is the minimum over schedules respecting
$M$ of the maximum machine load. For a threshold $T$, a configuration for
machine $i$ is a set $C$ of jobs each allowing $i$ with total size at most
$T$; the configuration linear program asks for weights $y_{i,C} \ge 0$ with
\[
\sum_{C} y_{i,C} = 1 \ \ \text{(every machine)}, \qquad
\sum_{(i,C)\,:\,j \in C} y_{i,C} = 1 \ \ \text{(every job)}.
\]
The first equality gives each machine total weight 1 on its configurations; the second covers each job with total weight 1. In addition, $y_{i,C} = 0$ unless $\mathrm{size}(C) \le T$ on $i$. Write $\OPTLP(I)$
for the least feasible $T$, and $\mathrm{gap}(I) = \OPT(I)/\OPTLP(I)$. An
instance with $\mathrm{gap}(I) \ge 3/2$ we call a \emph{witness}: it
certifies that the integrality gap of the class containing it is at least
$3/2$.

An instance with $|M(j)| \le 2$ for every job is a \emph{graph balancing} instance \cite{EKS14}: machines are vertices, jobs are edges, and a schedule orients each edge toward the machine that serves it. Both endpoints of an
edge carry the same size. We call the class with $|M(j)| \le 2$ and sizes
in $\{1,2\}$ \emph{two-weight graph balancing}.

The graph is a multigraph, and the multiplicity matters. An instance is a
job set, not a simple graph: several distinct jobs may share the same pair
of allowed machines, and $|M(j)| = 1$ is permitted, giving a job pinned to
one machine. Both features occur among the witnesses classified below: two of the three seven-job witnesses on four machines use one pair of machines three times. The exponents in Table~\ref{tab:seven} count such repeats, and a simple-graph reading misreads them.

Graph balancing is a strictly smaller model than the \emph{unrelated} graph balancing
of Verschae and Wiese \cite{VW14}, where an edge may carry two different
sizes, one per orientation. Their explicit construction for the
configuration linear program has gap $2 - 1/k$ with supremum 2 and lies in that larger
model (Figure~\ref{fig:landscape}). It therefore does not bear on the
class studied here.

\section{A six-job witness}

\begin{figure}[h]
\centering
\begin{tikzpicture}[scale=1.5]
  \node[circle,draw,inner sep=1.5pt] (v1) at (0,1) {$v_1$};
  \node[circle,draw,inner sep=1.5pt] (v2) at (1,1) {$v_2$};
  \node[circle,draw,inner sep=1.5pt] (v3) at (1,0) {$v_3$};
  \node[circle,draw,inner sep=1.5pt] (v4) at (0,0) {$v_4$};
  \draw[thick] (v1) -- node[above,font=\small] {1} (v2);
  \draw[thick] (v2) -- node[right,font=\small] {1} (v3);
  \draw[thick] (v3) -- node[below,font=\small] {1} (v4);
  \draw[thick] (v4) -- node[left,font=\small]  {1} (v1);
  \draw[very thick,dashed] (v1) -- node[pos=0.30,left,font=\small]  {2} (v3);
  \draw[very thick,dashed] (v2) -- node[pos=0.30,right,font=\small] {2} (v4);
\end{tikzpicture}
\caption{The witness $\Istar$: six jobs on four machines, integral
optimum 3 against relaxation value 2, so gap $3/2$. Machines are vertices;
jobs are edges labelled by size. Solid edges are the unit Hamiltonian
cycle; dashed edges are the weight-2 perfect matching.}
\label{fig:istar}
\end{figure}

\begin{theorem}\label{thm:istar}
Let $\Istar$ be the graph balancing instance on machines $v_1,\dots,v_4$
with unit jobs on the Hamiltonian cycle $v_1v_2v_3v_4v_1$ and weight-2 jobs
on the matching $\{v_1v_3, v_2v_4\}$ (Figure~\ref{fig:istar}). Then
$\OPT(\Istar) = 3$ and $\OPTLP(\Istar) = 2$, so $\mathrm{gap}(\Istar) = 3/2$.
\end{theorem}

\begin{proof}
(a) $T = 1$ is infeasible: a size-2 job exceeds the threshold, so no
configuration at any machine contains it.

(b) $T = 2$ is feasible (Figure~\ref{fig:frac}). Give every machine weight $1/2$ on the
configuration consisting of its matching job alone (load 2) and weight
$1/2$ on the configuration of its two cycle jobs (load $1 + 1 = 2$). Weights
sum to 1 per machine, and each job is covered $1/2 + 1/2 = 1$ by its two
endpoints.

(c) $\OPT(\Istar) \ge 3$ (Figure~\ref{fig:trap}). Every orientation of the perfect matching assigns a weight-2 job to some $x \in \{v_1,v_3\}$ and some $y \in \{v_2,v_4\}$. Both machines therefore carry load 2. Every pair in
$\{v_1,v_3\} \times \{v_2,v_4\}$ is a cycle edge, because the 4-cycle
alternates between the two matching classes. Hence the cycle job joining
$x$ and $y$ has both endpoints already at load 2, and whichever it is
assigned to reaches load 3. For the upper bound, a schedule of makespan 3
exists: send $v_1v_3$
and $v_1v_2$ to $v_1$ (load 3), $v_2v_4$ to $v_2$ (load 2), $v_2v_3$ and
$v_3v_4$ to $v_3$ (load 2), and $v_4v_1$ to $v_4$ (load 1).
\end{proof}

\subsection{Known results, and what is not ours}\label{sec:known}

$\Istar$ is not the first unconditional $3/2$ lower bound for this class. Jansen, Land, and Maack
\cite[Table~1]{JLM16} give a seven-job, four-machine instance and write:
``Even in this case the integrality gap is at least $3/2$, as the instance
given in Table 1 and its solutions depicted in Figure 1 show.'' That
instance has sizes in $\{1,2\}$ and every job on at most two machines, so it
lies inside two-weight graph balancing. The unconditional lower bound for
the class has therefore been in print since 2016. The matching upper bound
is in the same paper: \cite[Corollary~11]{JLM16} states that for two-valued
restricted assignment with sizes $s < b$ the configuration linear program's
integrality gap is at most $2 - s/b$, which is $3/2$ at $s = 1$, $b = 2$;
restricted assignment contains graph balancing, so the bound holds a
fortiori on our class. One paper, refereed, determines the value: the gap
for two-weight graph balancing is exactly $3/2$.

Two caveats belong with that statement. \cite{JLM16} state the algorithm
behind Corollary~11 but omit its analysis for space---``we omit the
analysis of the algorithm due to space restrictions''---and the corollary
itself, that the algorithm can be modified to work with the configuration
linear program, is asserted without proof. And an independent second route
to the same upper bound exists: Chakrabarty and Shiragur
\cite[remark following Theorem 1]{CS16} note that the integral optimum enters their
$3/2$-approximation analysis only through the feasibility of a flow, so a
relaxation feasible at threshold $T$ yields the same flow and a schedule of
makespan at most $3T/2$, bounding the relaxation's gap as well. It is a
remark in outline, in an unrefereed note, but its test (ratio against the
relaxation, not the integral optimum) is the right one. The same
approximation ratio was obtained independently and earlier by
Huang and Ott \cite{HO16} and by Page and Solis-Oba \cite{PSO16}; a remark
added to \cite{CS16} says so itself.

One more published seven-job witness lies in the \emph{broader} class. Wang and Sitters \cite[Section~4.1]{WS16} exhibit,
as an integrality-gap example for two-valued restricted assignment on
intervals: four machines, two tripled unit pairs, and one weight-2 job
allowed on \emph{all four} machines; gap $3/2$. The degree-4 job puts it
outside two-weight graph balancing, so it does not affect the class-scoped
claims here; but restricting that job's machine set to two machines, one
from each pair, yields exactly the third four-machine witness of
Table~\ref{tab:seven}. The primary-source review reached \cite{WS16} only
after our sweeps had produced that witness independently; the kinship is
noted where the classification is stated.

The rest of this note is therefore about the \emph{size} of a witness, not the value of the gap.

\paragraph{Relation to the published witness.}
Table~1 of \cite{JLM16} gives a seven-job, four-machine instance with gap
$3/2$: one size-2 job on a path anchored by two singleton unit jobs, each
allowed on one machine only ($|M(j)| = 1$).
$\Istar$ is not that instance with a job deleted. Counting each singleton job as one incidence at its one allowed machine, both instances have three job incidences at each of their four machines, so incidence counts cannot distinguish them. But $\Istar$ has two size-2 jobs rather than one, ruling out a size-preserving embedding under relabeling. Also, $\Istar$ is the simple graph $K_4$ on all six machine pairs, while the published instance uses three distinct pairs with two parallel classes. An exhaustive
search over all $4!$ machine relabelings and all injective size-preserving
job maps confirms no embedding exists.

\subsection{Exact certification}

Both $\Istar$ and the instance of \cite{JLM16} are certified independently
of any solver. Infeasibility at $T = 1$ follows structurally. Feasibility at
$T = 2$ is checked in exact rational arithmetic: weights summing to 1 per
machine, every job covered exactly 1, every used configuration's load at
most $T$ by integer comparison, and---the constraint most easily
omitted---every job in a configuration admissible on that machine. Script:
\texttt{p29\_exact\_rational\_verify.py}. For \cite{JLM16} we supply our own
symmetric witness rather than transcribing theirs; any valid witness proves
feasibility.

\section{Minimality and uniqueness}

\begin{lemma}[connectivity reduction]\label{lem:conn}
Model an instance as a multigraph on machines, jobs as edges, a degree-1 job
as a pendant edge to a dedicated vertex. Configuration-linear-program
feasibility and $\OPT$ both decompose over connected components---this is
Lemma~\ref{lem:decomp}, proved below and independently of this one---so
$\OPT(I)$ and $\OPTLP(I)$ are the maxima over components. Consequently some
component $c$ has $\mathrm{gap}(c) \ge \mathrm{gap}(I)$: taking $c$ to
attain $\OPT(I)$ gives
$\mathrm{gap}(c) = \OPT(I)/\OPTLP(c) \ge \OPT(I)/\OPTLP(I)$. A connected
multigraph with $k$ edges has at most $k+1$ vertices, so any witness with at
most $k$ jobs contains a connected witness on at most $k+1$ machines: at
most six machines for the five-job minimality claim, at most seven for the
six-job uniqueness claim.
\end{lemma}

The pendant-vertex device is a counting convenience for bounding vertices
only. The dedicated vertex is not an allowed machine for the job---were it one, the job could go there without loading any real machine and both optima would
change---so it carries no load and alters neither $\OPT$ nor $\OPTLP$.

\begin{lemma}[weight bound]\label{lem:weight}
If the configuration linear program is feasible at threshold $T$ on $m$
machines, the total weight satisfies $\sum_j p_j \le mT$: every job's
coverage constraint weights it exactly once, so
$\sum_j p_j = \sum_{i,C} y_{i,C}\,\mathrm{load}(C) \le T \sum_{i,C} y_{i,C}
= mT$, using each machine's unit weight budget. Since sizes are at least 1,
an instance with $\OPTLP = T$ has at most $mT$ jobs.
\end{lemma}

\begin{theorem}[minimality across all machine counts]\label{thm:min}
Within two-weight graph balancing, no instance with at most five jobs has
gap at least $3/2$, on any number of machines. $\Istar$ is therefore a minimum
witness; the seven-job instance of \cite{JLM16} is not.
\end{theorem}

\begin{proof}
By Lemma~\ref{lem:conn} it suffices to check witnesses on at most six
machines. The sweep at six machines and at most five jobs (433{,}743
feasibility decisions) found none, and this subsumes every smaller machine
count, since a machine no job can use affects neither $\OPT$ nor $\OPTLP$.
Sweeps at four machines (34{,}200) and five (139{,}975) agree; Figure~\ref{fig:wall} shows the per-size decision counts. Every
decision is exact; see Section~\ref{sec:verify}.
\end{proof}

\begin{theorem}[uniqueness at six jobs, all machine counts]\label{thm:unique}
Within two-weight graph balancing, on any number of machines, the instances
with six jobs and gap $3/2$ are exactly the copies of $\Istar$, possibly
together with machines no job can use. Up to machine relabeling and such
unused machines, $\Istar$ is the unique witness at its size.
\end{theorem}

\begin{proof}
First reduce the machine count. A six-job witness has, by
Lemma~\ref{lem:conn}, a connected component of gap at least $3/2$ with at
most six jobs; Theorem~\ref{thm:min} rules out five or fewer, so that
component carries all six jobs, and the rest of the instance is machines no
job can use. Six edges and connectivity bound the component's machines by
seven, so sweeps at up to seven machines decide the theorem for every
machine count.

At four machines the redundant enumeration is feasible: 192{,}847
feasibility decisions return exactly three instances of gap
$3/2$. In each, the two weight-2 jobs form one of the three perfect
matchings of $K_4$---a different one in each instance---and the four unit
jobs occupy four \emph{distinct} machine pairs, leaving every vertex at
degree 2; four distinct edges that are 2-regular on four vertices form the
4-cycle complementary to the matching, so no parallel-edge configuration
arises. The three perfect matchings of $K_4$ are equivalent under
relabeling (Figure~\ref{fig:orbit}), so the three instances are isomorphic copies of $\Istar$. The
count is corroborated by the orbit--stabilizer theorem: $\Istar$'s
stabilizer among the $4! = 24$ machine relabelings is the dihedral group of
its 4-cycle, of order 8, so its orbit has size $24/8 = 3$, exactly the
number of hits. Because the enumeration retains machine labels, a single witness appears as three isomorphic copies. This sweep covers
every smaller machine count as well, since it includes instances that leave
machines unused.

At five, six, and seven machines the redundant enumeration is out of
reach---75{,}693{,}820 instances across the three ranges---so there we
enumerate one representative per machine-relabeling class instead, a
\emph{canonical} enumeration, the term used for such sweeps below:
22{,}951,
32{,}377, and 37{,}316 classes, of which 15{,}645, 17{,}901, and 18{,}427
survive pruning to a feasibility decision. Each range returns exactly one
hit: $\Istar$ on four of the machines, the others untouched. No witness
whose six jobs genuinely use a fifth machine exists. The class counts of
this enumeration are confirmed against Burnside's lemma on every range, and
its hit set at four machines against the redundant sweep's; both checks are
described in Section~\ref{sec:verify}.
\end{proof}

\noindent
The six- and seven-machine ranges need not rest on the enumerator at all:
they follow from the elementary argument Theorem~\ref{thm:nine} uses on its
own tail. A connected six-job class on $m$ machines with $s$ singleton jobs
has $6 = E \ge V - 1 = m + s - 1$, so $s \le 7 - m$: at $m = 7$ there is no
dedicated vertex, and $V = 7$ against $E = 6$ makes the multigraph a tree;
at $m = 6$ there is at most one, and the multigraph is a tree if a singleton
job occurs and carries a single cycle if none does. Orient a tree away from
a root---the dedicated vertex where there is one, so that its pendant edge
points into its own machine---and a unicyclic multigraph consistently around
its cycle and away from it elsewhere; either way no machine has in-degree
above one, so $\OPT \le p_{\max} = 2$, while a witness has $\OPT = 3$ by
Corollary~\ref{cor:cell}. Those two ranges are therefore corroborated by
their sweeps rather than carried by them, and the five-machine range is the
only one the canonical enumerator decides alone.

At seven jobs the number of witness classes rises from one to thirteen. The instance of \cite{JLM16} turns out
to be one of exactly three seven-job witnesses on four machines. Of the
other two, one is obtained from Wang and Sitters' broader-class example
by restricting its weight-2 job to two machines
(Section~\ref{sec:known})---a different instance, in a class where
restriction is not known \emph{a priori} to raise or lower the gap---and
neither it as such, nor the third, nor any of the six witnesses on five
and six machines, appears to have been published.

\begin{theorem}[the seven-job witnesses]\label{thm:seven}
Up to machine relabeling and machines no job can use, the instances of
two-weight graph balancing with seven jobs and gap $3/2$, on any number of
machines, are exactly thirteen. Nine are connected: three on four machines
---among them the instance of \cite[Table~1]{JLM16}---five on five
machines, and one on six. The remaining four are $\Istar$ together with
one job on machines disjoint from its four: a size-1 or size-2 job,
allowed on one machine or on two. Table~\ref{tab:seven} lists all
thirteen.
\end{theorem}

\begin{proof}
A seven-job witness has, by Lemma~\ref{lem:conn}, a connected component of
gap at least $3/2$ with at most seven jobs, and Theorem~\ref{thm:min}
forces that component to carry six or seven of them. If six, the component
is $\Istar$ (Theorem~\ref{thm:unique}) and the remaining job is a
component of its own, on machines disjoint from $\Istar$'s; conversely
every such instance is a witness, because both optima decompose over
components: a single job of size $p \le 2$ gives
$\OPT = \max(3, p) = 3$ and $\OPTLP = \max(2, p) = 2$. That yields the
four disconnected classes. If seven, the component is connected with
seven edges, hence on at most eight machines, so the canonical sweeps at
four through eight machines and at most seven jobs decide the connected
case: $39{,}500$, $90{,}796$, $130{,}508$, $145{,}339$, and 149{,}073
feasibility decisions respectively, with class counts again matching
Burnside's lemma on every job count of every range. Collapsing the hits by
the machines they touch yields exactly the thirteen classes of
Table~\ref{tab:seven} and no others; every class supported on $k$ machines
reappears at every swept $m \ge k$ and at no smaller $m$, and no hit
touches seven or eight machines. The identification of
\cite[Table~1]{JLM16} among the four-machine hits is by explicit
relabeling, checked mechanically: a positive control the pipeline was
not designed around.
\end{proof}

\begin{table}[htbp]
\centering
\small
\begin{tabular}{cccp{0.52\textwidth}}
\hline
machines & weight & kind & jobs (allowed set$^{\text{size}}$, exponents = multiplicity) \\
\hline
4 & 8 & connected &
  $\{a\}^1\, \{ac\}^{1\times 2}\, \{b\}^1\, \{bd\}^{1\times 2}\, \{cd\}^2$
  \quad (= \cite[Table~1]{JLM16}) \\
4 & 8 & connected &
  $\{a\}^1\, \{ad\}^{1\times 2}\, \{bc\}^{1\times 3}\, \{cd\}^2$ \\
4 & 8 & connected &
  $\{ac\}^{1\times 3}\, \{bd\}^{1\times 3}\, \{cd\}^2$ \\
5 & 9 & connected &
  $\{a\}^1\, \{ab\}^{1\times 2}\, \{bc\}^2\, \{cd\}^1\, \{ce\}^1\, \{de\}^2$ \\
5 & 9 & connected &
  $\{ab\}^{1\times 3}\, \{bc\}^2\, \{cd\}^1\, \{ce\}^1\, \{de\}^2$ \\
5 & 9 & connected &
  $\Istar$ on $\{b,c,d,e\}$ plus a size-1 job on $\{a,b\}$ (pendant) \\
5 & 10 & connected &
  $\Istar$ on $\{a,b,c,d\}$ plus a size-2 job on $\{d,e\}$ (pendant) \\
5 & 10 & connected &
  $\{ab\}^1\, \{ac\}^1\, \{ad\}^2\, \{bd\}^1\, \{be\}^2\, \{cd\}^1\, \{ce\}^2$
  \quad (one diagonal of $\Istar$ subdivided) \\
6 & 10 & connected &
  $\{ab\}^2\, \{ac\}^1\, \{ad\}^1\, \{be\}^1\, \{bf\}^1\, \{cd\}^2\, \{ef\}^2$
  \quad (two triangles, each with one size-2 side, joined by a size-2 job) \\
5 & 9 & disconnected & $\Istar \oplus$ a size-1 job on one machine \\
5 & 10 & disconnected & $\Istar \oplus$ a size-2 job on one machine \\
6 & 9 & disconnected & $\Istar \oplus$ a size-1 job on two machines \\
6 & 10 & disconnected & $\Istar \oplus$ a size-2 job on two machines \\
\hline
\end{tabular}
\caption{The thirteen seven-job witnesses, up to machine relabeling and
machines no job can use. Machines are $a, b, \dots$; weight is total job
size; every class has integral optimum 3 against relaxation value 2. The
machine column is the size of the support. Generated mechanically from the
sweep artifacts by \texttt{p29\_seven\_job\_classify.py}.}
\label{tab:seven}
\end{table}

By Lemma~\ref{lem:weight} a seven-job witness at threshold 2 on four
machines has weight at most 8, so it carries at most one size-2 job; all
three four-machine classes saturate the bound exactly. The five- and
six-machine classes show the bound is not always tight: four of the
five-machine classes have weight 9 against a budget of 10.

The next bound is not ours, and two of the results below rest on it, so we
state it carefully before using it. The two are Corollary~\ref{cor:cell} and, through it, the short proof of
Theorem~\ref{prop:mach}. The minimality and uniqueness
theorems use only Lemma~\ref{lem:conn} and the sweeps, and
Theorem~\ref{prop:mach} carries a self-contained fallback that needs
neither. Lenstra,
Shmoys and Tardos \cite{LST90} round a feasible assignment-relaxation
solution at target $T$ to a schedule of makespan at most $T + p_{\max}$.
Verschae and Wiese sharpen the inequality to a strict one
\cite[p.~375]{VW14}, remarking that in \cite{LST90} ``the inequality in the
theorem is not strict'' but that ``it is easy to see that the same proof
yields a strict inequality.'' They give no proof of it, there or anywhere
else in the paper---their own Theorem~4 uses the strict form, as
Corollary~\ref{cor:cell} does---so we prove it. For two-valued instances the additive form is also
stated by Jansen, Land and Maack, whose algorithm ``has an additive
approximation guarantee of $\OPTLP + b - s$'' \cite{JLM16}; at $s=1$, $b=2$
that is $\OPT \le \OPTLP + 1$.

\begin{lemma}[strict rounding, after \cite{LST90,VW14}]\label{lem:round}
If the configuration linear program is feasible at an integer threshold $T$,
then $\OPT \le T + p_{\max} - 1$, where $p_{\max}$ is the largest job size.
For the class this reads $\OPT \le T + 1$, at every number of machines.
\end{lemma}

\begin{proof}
Feasibility at $T$ gives fractional shares $x_j^i = \sum_{C \ni j} y_{i,C}$
covering each job once with fractional load at most $T$ per machine, and
$x_j^i = 0$ unless $i \in M(j)$. A job with $p_j > T$ lies in no configuration, so its coverage constraint could not be met; hence $p_j \le T$ for every job, a hypothesis enforced by the \cite{LST90} relaxation beyond the plain assignment relaxation: their $x_{i,j} = 0$ whenever $p_{i,j} > T$, as \cite[p.~375]{VW14} render it.

Pass to a vertex of the \emph{assignment} polytope these constraints cut
out, not of the configuration linear program's feasible set; that set is
strictly smaller, as the display in Section~\ref{sec:t2} shows. The support
only shrinks, so sizes in it are still at most $p_{\max}$. At a vertex the
bipartite support graph on machines and jobs is a \emph{pseudoforest}---%
their word: every component is ``either a tree or a tree plus one
additional edge''---and deleting the integrally assigned jobs leaves every
remaining job node of degree at least two, so it carries a matching
covering all of them \cite[Theorem~1]{LST90}; rounding
along it gives each machine at most one fractionally assigned job. In the
graph balancing model this is \cite[\S2.2]{EKS14}, where the fractional
edges of a basic solution form a forest and are oriented one-to-one away
from a root. Fix a machine $i$ and
let $j_0$ be that job, if it has one. The jobs assigned to $i$ integrally
carry load at most $T - p_{j_0} x_{j_0}^i < T$, since $p_{j_0} > 0$ and
$x_{j_0}^i > 0$; both that load and $T$ are integers, so it is at most
$T - 1$, and $i$ finishes at most $T - 1 + p_{j_0} \le T + p_{\max} - 1$. A
machine receiving no fractional job finishes at most $T$.
\end{proof}

Three remarks. First, the vertex hypothesis is needed: at an arbitrary feasible
$x$ the support graph need not be a pseudoforest and a machine may carry
several fractional jobs. The vertex hypothesis controls the shape of the support, not its largest
size: every job has a positive share somewhere, so the largest size in the
support is always $p_{\max}$. Second,
Shchepin and Vakhania \cite{SV05} round a relaxation of the same family to
$T + \tfrac{m-1}{m}p_{\max}$, which is stronger than Lemma~\ref{lem:round} exactly
when $p_{\max} > m$ and matches it here after integrality. They call that
additive term best possible for the rounding approach, but the scope is
narrower than the phrase suggests: their argument is a closing paragraph of
\cite[\S4]{SV05} whose witness is a single job on $m$ identical machines,
and it measures a non-preemptive optimum against an optimal preemptive
distribution. On that instance the configuration linear program is feasible
exactly at $T = \OPT$, so it has no additive gap at all, and the witness
says nothing about Lemma~\ref{lem:round}. \cite[p.~372]{VW14} state the
scope correctly: best possible ``among all rounding algorithms for this
LP''. Sharpening the rounding of the configuration linear program is
therefore not closed off; it is simply not what \cite{SV05} address.
Third, without the strictness the search doubles: the non-strict bound gives
$\OPT \le T + 2$, hence $\OPTLP \le 4$ below, hence twelve jobs to check
instead of six.

\begin{corollary}\label{cor:cell}
Every witness of the class has $(\OPTLP, \OPT) = (2,3)$. In particular no
instance of the class has $\OPTLP = 4$ and $\OPT = 6$, at any number of
machines.
\end{corollary}

\begin{proof}
By Lemma~\ref{lem:round}, $\mathrm{gap} \le 1 + 1/\OPTLP$, so
$\mathrm{gap} \ge 3/2$ forces $\OPTLP \le 2$. If $\OPTLP = 1$ then no job has
size 2---a size-2 job has no admissible configuration at threshold 1---so
$p_{\max} = 1$ and the lemma gives $\OPT \le 1$, a gap of 1. Hence
$\OPTLP = 2$, and then $\OPT \le 3$ with $\mathrm{gap} \ge 3/2$ forces
$\OPT = 3$.
\end{proof}

The corollary depends on sizes being integers in $\{1,2\}$, not on
any structural feature of graph balancing: doubling $\Istar$ to sizes
$\{2,4\}$ gives an instance with $\OPTLP = 4$ and $\OPT = 6$ exactly. The
cell is empty for this class, not for the model.

\begin{theorem}[three machines never suffice]\label{prop:mach}
No three-machine instance of the class attains gap at least $3/2$, at any
number of jobs. Four machines are therefore necessary for the $3/2$ gap,
and $\Istar$ uses exactly four: it is a minimum witness in machines as
well as in jobs.
\end{theorem}

\begin{proof}
By Corollary~\ref{cor:cell} a witness has $\OPTLP = 2$, so by
Lemma~\ref{lem:weight} its total weight is at most $3 \cdot 2 = 6$; sizes
are at least 1, so it has at most six jobs. The canonical sweep at three
machines and at most six jobs decides $3{,}237$ feasibility questions and
finds no witness, in floating point and again in exact rational arithmetic
with identical counts and identical (empty) hit sets. That proves the
theorem. The swept range is deliberately wider than the argument needs:
six jobs admit weight up to 12, so the check is stronger than the bound
requires.

The theorem also follows from the elementary argument below together with a wider sweep, without Corollary~\ref{cor:cell} and hence without the strict form of Lemma~\ref{lem:round}. It is kept because it
borrows nothing from the literature: Lemma~\ref{lem:round} relies on the rounding argument of
\cite{LST90,EKS14}, while the route below uses only subset sums and a
balanced orientation. The redundant sweep
at three machines and at most nine jobs (292{,}971 decisions) and the
canonical sweep extending it to twelve (459{,}698 decisions over
459{,}889 classes, Burnside-checked) find no witness, covering every
$T \le 4$ by Lemma~\ref{lem:weight}; every decision is exact, see
Section~\ref{sec:verify}.

For $T \ge 5$ the exclusion is then analytic. Feasibility at $T$ yields
fractional shares $f_{j,i} = \sum_{C \ni j} y_{i,C}$ with unit coverage
per job and fractional load at most $T$ per machine. On machines
$\{a,b,c\}$ every job is a singleton---forced to its machine---or lies in
one of the three pair pools $P_{ab}$, $P_{bc}$, $P_{ca}$. Write
$\alpha^i_P$ for pool $P$'s fractional weight on machine $i$; then
$S_a + \alpha^a_{ab} + \alpha^a_{ca} \le T$ and cyclically, where $S_i$
is machine $i$'s singleton load. Now round each pool separately. The
subset sums of a multiset with sizes in $\{1,2\}$ contain $0$ and the
total and have consecutive gaps at most $2$, so each pool admits an
integral split within $1$ of its fractional one; assign accordingly, and
every singleton to its machine. Each machine meets exactly two pools, so
its integral load is at most $T + 2$. Hence $\OPT \le T + 2$ whenever
the relaxation is feasible at $T$, and
$\mathrm{gap} \le 1 + 2/T \le 7/5 < 3/2$.
\end{proof}

In that second route the two regimes meet exactly: at $T = 4$ the analytic
bound gives $\OPT \le 6$, which forbids any gap \emph{above} $3/2$ but
still permits a witness at exactly $3/2$. The sweep is therefore required to
rule that case out, not merely to overlap the analytic regime; its reach of
$n \le 12 = 3 \cdot 4$ jobs is precisely the weight bound at $T = 4$. The generic
form of the rounding on $m$ machines gives $\OPT \le T + (m-1)$; a sharper,
oriented form of the same argument tightens both small machine counts by one
unit. Lemma~\ref{lem:round} supersedes all of these---it gives $\OPT \le T+1$
at every machine count---so what follows is not a new bound but an elementary one: it uses no rounding theorem from the literature, whereas Lemma~\ref{lem:round} borrows the pseudoforest-and-matching step from \cite{LST90,EKS14}.

\begin{proposition}[oriented pool rounding]\label{prop:orient}
For instances of the class with the relaxation feasible at threshold
$T$: on three machines $\OPT \le T + 1$, and on four machines
$\OPT \le T + 2$ (Figure~\ref{fig:pools}).
\end{proposition}

\begin{proof}
Round each pair pool to an achievable subset sum nearest its fractional
split, as in Theorem~\ref{prop:mach}, but account for the excess by
cases. A pool containing a unit job achieves every integer up to its
weight, so one endpoint takes the ceiling of its fractional share
(excess less than $1$) and the other the floor (excess at most $0$),
and \emph{which} endpoint takes the ceiling is free. An all-weight-2
pool achieves exactly the even integers; the even integer nearest one
endpoint's share deviates by at most $1$, and the deviations at the two
endpoints are opposite, so each such pool contributes, to one endpoint
only: exactly $+1$ when the shares split at odd integers (and which endpoint bears it is free), an excess strictly between $0$ and $1$ on
a \emph{forced} endpoint when the shares are fractional, and nothing
when they split at even integers.

Fix a machine $v$ meeting $\deg(v)$ pools, and let $c_1$ count the
exact-$+1$ pools directed into $v$ and $k$ the strictly-fractional
positive excesses at $v$. The machine's integral load is an
integer at most $T$ plus the total excess, and the $k$ fractional
contributions sum to strictly less than $k$, so the load is at most
$T + c_1 + \max(k-1, 0)$. Since $c_1 + k \le \deg(v)$, the bound can
exceed $\deg(v) - 1$ only when $k = 0$ and $c_1 = \deg(v)$: every pool
at $v$ an odd all-weight-2 pool, all directed in. Directing the odd
all-weight-2 pools along a \emph{balanced} orientation of their subgraph
(one in which in-degree and out-degree differ by at most one at every
vertex, which every graph admits) caps their in-degree at
$\lceil \deg/2 \rceil$---at most $1$ on three
machines and $2$ on four---which prevents that case. Hence the excess
is at most $1$ on three machines ($\deg \le 2$) and $2$ on four
($\deg \le 3$).
\end{proof}

Even without Lemma~\ref{lem:round}, Proposition~\ref{prop:orient}
retroactively simplifies the elementary route to
Theorem~\ref{prop:mach}: on three machines it excludes every $T \ge 3$
outright ($\lceil 3T/2 \rceil > T+1$ there), $T = 1$ forbids weight-2
jobs and rounds integrally, and $T = 2$ instances carry weight at most
$6$, so the twelve-job sweep is needed only as independent assurance.
On four machines it leaves $T \in \{2,4\}$: gap $3/2$ with integral $\OPT$
and $\OPT \le T + 2$ rule out every integer $T \ge 5$, while odd
$T \in \{1,3\}$
would force gap strictly above $3/2$, contradicting the ceiling of
Section~\ref{sec:known}. Corollary~\ref{cor:cell} then removes $T = 4$ as
well, so at four machines, as at every other count, a witness has
$\OPTLP = 2$ (Figure~\ref{fig:frontier}), where the witnesses are
classified (Theorems~\ref{thm:unique}, \ref{thm:seven},
\ref{thm:eight}). The cell $(T,\OPT) = (4,6)$, which the oriented bound
alone does not exclude---there it degenerates to exactly $3/2$, and it is
where the parity of all-weight-2 pools can no longer be rounded away---is
empty for this class by the corollary.

\begin{theorem}[the eight-job witnesses, counted]\label{thm:eight}
Up to machine relabeling and machines no job can use, exactly $154$
instances of the class with eight jobs have gap $3/2$: $97$ connected
---$41$ supported on five machines, $50$ on six, $6$ on seven, none
needing eight---and $57$ disconnected, each a six- or seven-job witness
together with disjoint jobs whose components have relaxation value at
most 2. The witness count grows $1$, $13$, $154$ at six, seven, and
eight jobs.
\end{theorem}

\begin{proof}
The reduction of Theorem~\ref{thm:seven} applies verbatim one level up:
an eight-job witness's gap-attaining component carries six, seven, or
eight jobs; the first two cases are Theorems~\ref{thm:unique}
and~\ref{thm:seven} plus disjoint jobs, and a connected eight-edge
component fits on at most nine machines. Canonical sweeps at four
through nine machines and at most eight jobs decide the rest:
$138{,}950 + 446{,}708 + 846{,}081 + 1{,}111{,}144 + 1{,}208{,}344 +
1{,}233{,}378 = 4{,}984{,}605$ feasibility decisions, class counts
matching Burnside's lemma on every job count of every range, every
decision made twice with identical hit sets. Collapsing the hits by
support yields exactly $154$ classes, consistent across all six ranges.
Two independent corroborations: the $57$ disconnected classes coincide
class-for-class with the composite calculus computed by hand from
Theorems~\ref{thm:unique} and~\ref{thm:seven} \emph{before} the sweeps
ran, and no hit's support exceeds eight machines, one below the
connectivity bound. The class list is deposited
(\texttt{seven\_job\_classification\_n8.json}); we do not reproduce
154 rows here.
\end{proof}

\noindent
The same split---an exhaustive sweep for the connected classes, a composite calculus for the rest---reaches nine jobs. Two ingredients determine the calculus, and we name them first.

\begin{definition}\label{def:gapcomp}
A \emph{gap component} of an instance is a connected component that is
itself a witness. A \emph{padding component} is a connected component
that is not a gap component and has $\OPTLP \le 2$. The exclusion is not
cosmetic: every witness has $\OPTLP = 2$ by Corollary~\ref{cor:cell}, so
without it every gap component would also be a padding component and the
two would not partition the components of a witness.
\end{definition}

\begin{lemma}[the relaxation decomposes]\label{lem:decomp}
Let $I$ be the disjoint union of instances $c_1, \dots, c_r$ on disjoint
machine sets. Then $\OPT(I) = \max_t \OPT(c_t)$ and
$\OPTLP(I) = \max_t \OPTLP(c_t)$.
\end{lemma}

\begin{proof}
For $\OPT$: machines are partitioned among the components and every job's
allowed set lies inside one component, so assignments correspond to tuples
of per-component assignments and the makespan is the maximum of theirs.

For $\OPTLP$ the point is that \emph{no column of the relaxation mixes
components}. A configuration for machine $i$ is a set of jobs allowed on
$i$; every such job has $i$ in its allowed set, hence lies in $i$'s
component. So each variable $y_{i,C}$ involves one component only, and at
any threshold $T$ the constraint matrix is block diagonal, the machine
rows and job rows of a component meeting that component's columns and no
others. Feasibility at $T$ therefore holds for $I$ if and only if it holds
for every $c_t$: restrict a solution to a block for one direction,
concatenate block solutions for the other. Restriction preserves the
machine equality intact, since every column meeting a machine of $c_t$ is
a column of $c_t$, so the sum $\sum_C y_{i,C} = 1$ is unchanged by
discarding the other blocks.

Passing to least thresholds needs two remarks, monotonicity and non-emptiness. Feasibility is \emph{monotone} in $T$: a configuration
admissible at $T$ is admissible at $T+1$, so a solution extends by zeros.
Hence each $c_t$ is feasible exactly on an up-set
$\{T : T \ge \OPTLP(c_t)\}$, and $I$ is feasible on their intersection,
which is $\{T : T \ge \max_t \OPTLP(c_t)\}$. The intersection is
non-empty because $T = \OPT(I)$ is feasible, an integral schedule being a
$0/1$ solution. Its least element is $\max_t \OPTLP(c_t)$, which is the
claim.

One direction needs stating separately: $I$ cannot be feasible at a $T$ where
some $c_t$ is infeasible. It cannot, precisely because restriction is
available: with no shared columns there is no slack to borrow from another
block.
\end{proof}

\begin{theorem}[the nine-job witnesses, counted]\label{thm:nine}
Up to machine relabeling and machines no job can use, exactly $1{,}662$
instances of the class with nine jobs have gap $3/2$: $978$ connected
---$49$ supported on five machines, $489$ on six, $387$ on seven, $53$ on
eight, none on four and none on nine or more---and $684$ disconnected,
each a unique gap component of six, seven, or eight jobs together with a
multiset of padding components carrying the remaining three, two, or one
job. The witness count grows $1$, $13$, $154$, $1{,}662$ at six through
nine jobs.
\end{theorem}

\begin{proof}
\emph{Strategy}: the connected classes are swept exhaustively; the
disconnected ones are counted by showing that every one of them
decomposes uniquely as a gap component plus a padding multiset, and then
enumerating those.

\emph{Connected.} The sweep of Theorem~\ref{thm:eight} extended to nine
jobs, sharded by support size $k = 4, \dots, 10$, returns
$0 / 49 / 489 / 387 / 53 / 0 / 0$ hits. Supports nine and ten contain no
witness, and this is elementary rather than a fact about the sweep: a
connected class with nine jobs on nine or ten machines has $E \le V$, so
its multigraph is a tree or carries a single cycle. It has at most one
dedicated vertex, and on ten machines none: a singleton job contributes a
vertex as well as an edge, so connectivity of a class on $m$ machines with
$s$ singleton jobs requires $9 = E \ge V - 1 = m + s - 1$, forcing
$s \le 1$ at $m = 9$ and $s = 0$ at $m = 10$. Orienting that cycle
consistently, and each pendant tree away from it---rooting at the
dedicated vertex of Lemma~\ref{lem:conn} when a singleton job creates
one, so that its pendant edge is oriented into its own machine and not
onto a machine the job may not use---gives in-degree at most
one at every machine; assigning each job to the machine its edge points
into puts at most one job on a machine, so $\OPT \le p_{\max} = 2$, while
a witness has $\OPT = 3$ by Corollary~\ref{cor:cell}. (The sweep's own
artifacts record no feasibility decision at those supports, every class
there falling below the prune threshold, but that is a property of the
pruning and is not what makes the cells empty.) The remaining
five shards make $4{,}231{,}656$ decisions, each made twice, in floating
point and in exact rational arithmetic, agreeing on the hit
\emph{sets} and not merely on their sizes; all $978$ hits were then
decided a third time by a separately written implementation produced
during adversarial review and sharing no code with the sweep
(\texttt{p29\_referee\_probes/}), with no mismatch.
That third pass re-decides recorded hits and so cannot detect a hit the
sweep missed; completeness rests on the enumeration, not on it.

\emph{Disconnected, and why the calculus is forced.} Let $I$ be a
disconnected nine-job witness with components $c_1, \dots, c_r$, and let
$c^*$ attain $\max_t \OPT(c_t)$. By Lemma~\ref{lem:decomp},
\[
  \gap(I) \;=\; \frac{\OPT(c^*)}{\OPTLP(I)}
           \;\le\; \frac{\OPT(c^*)}{\OPTLP(c^*)} \;=\; \gap(c^*),
\]
so $\gap(I) \ge 3/2$ forces $\gap(c^*) \ge 3/2$: some component is itself
a witness, hence a gap component, hence carries at least six jobs by
Theorem~\ref{thm:min}. Two gap components would need at least twelve, so
there is exactly one, and it is recoverable from $I$ without choice as the
unique component with six or more jobs. It carries at most eight: $I$ has
at least two components, and every component carries at least one job.

For the padding condition, apply Corollary~\ref{cor:cell} to $I$ itself.
$I$ is a witness, so $\OPTLP(I) = 2$, and Lemma~\ref{lem:decomp} then
gives $\OPTLP(c) \le 2$ for \emph{every} component $c$: the condition is
derived rather than imposed. The cores' profiles alone do not imply the padding condition. Since $\OPT(I)$ is the \emph{maximum} over components, knowing
every core is at $(\OPTLP, \OPT) = (2,3)$ gives only $\OPT(I) \ge 3$, not
equality; a component on at most three jobs can have $\OPT$ as large as
$6$. An argument avoiding Corollary~\ref{cor:cell}, and so independent of
the strict rounding lemma, is available once the cores' optima are taken
from the deposited classifications rather than from that corollary. The
$1 + 9 + 97$ cores are all at $\OPT = 3$ there: if some non-core component had
$\OPTLP \ge 3$ then $\OPTLP(I) \ge 3$, so $\gap(I) \ge 3/2$ would force
$\OPT(I) \ge 5$; that exceeds the core's $3$, so the maximum is attained
at a non-core component $R$, and then $\gap(R) \ge \OPT(I)/\OPTLP(I) \ge
3/2$ makes $R$ a witness on at most three jobs, contradicting
Theorem~\ref{thm:min}. Conversely any such union is a witness, by the
same lemma.
Finally, machines carrying no job are identified away, so the multiset of
components is an isomorphism invariant and matching multisets compose into
an isomorphism; distinct multisets therefore give non-isomorphic
instances, and since the core is recoverable, distinct (core, multiset)
pairs give distinct instances.

\emph{Counting.} Writing $p_k$ for the number of padding classes on $k$
jobs and $P(k)$ for the number of nonempty multisets of padding components
totalling $k$ jobs, exhaustive enumeration at one, two and three jobs
gives $4$, $13$ and $60$ connected classes of which $p_1 = 4$,
$p_2 = 11$ and $p_3 = 43$ are padding, whence $P(1) = 4$,
$P(2) = 11 + 10 = 21$ and $P(3) = 43 + 44 + 20 = 107$. With $97$, $9$ and
$1$ cores at eight, seven and six jobs,
\[
  D(9) \;=\; 97 P(1) + 9 P(2) + 1 P(3)
       \;=\; 388 + 189 + 107 \;=\; 684 ,
\]
and $978 + 684 = 1{,}662$. Figure~\ref{fig:calculus} draws the decomposition and the
three products.
\end{proof}

\noindent
\emph{Restatement check.} The theorem claims an exact count of
isomorphism classes at nine jobs, and the proof delivers one only because
the two halves are disjoint by construction---connected against
disconnected---and each is exhaustive on its own side. Three checks support the count, and each could have failed. The calculus was run at
seven and eight jobs, where the answer is already known from
Theorems~\ref{thm:seven} and~\ref{thm:eight}: it returns $D(7) = 4$ and
$D(8) = 57$, and the classes it constructs coincide with the deposited
classification \emph{as sets}, with empty symmetric difference, not merely
in number. All $684$ predicted nine-job instances were then assembled and
their $\OPT$ and $\OPTLP$ recomputed \emph{directly on the assembled
instance}---not as a maximum over components---confirming that all $684$
are witnesses with nine jobs and are pairwise non-isomorphic; so the count
is a lower bound by explicit construction, independently of
Lemma~\ref{lem:decomp}, and the lemma is what makes it an upper bound.
And $p_1$, $p_2$, $p_3$ and the connected class counts $4$, $13$, $60$
were recomputed by a second implementation sharing no code with the first
(\texttt{p29\_padding\_recount.py}), with every feasibility decision
\emph{decided} in exact rational arithmetic. We say decided rather than
certified deliberately: that script returns a verdict, not a Farkas
vector, and this note keeps the two words apart.

The case that would break the argument is a witness union containing a
\emph{non-padding} component, one with $\OPTLP \ge 3$. That is the case
that threatens completeness, because the calculus enumerates padding
multisets only and would never build such a union. The preceding results exclude it: as shown above, $\OPTLP(I) = 2$ for any witness $I$ by
Corollary~\ref{cor:cell}, so Lemma~\ref{lem:decomp} leaves no component
above $2$; and independently of that corollary, a non-core component with
$\OPTLP \ge 3$ would force $\OPT(I) \ge 5$ and thereby produce a witness
on at most three jobs, contradicting Theorem~\ref{thm:min}.

A neighbouring case would not disturb the count, and it does not occur either. A padding component with $\OPT \ge 4$ would raise the
union's gap to $2$, above the $3/2$ the theorem records; but it would
\emph{not} disturb the count, since such a component is padding and the
union is enumerated as core plus padding multiset like any other. It too
fails to arise, since such a component would be a witness on at most
three jobs. At these sizes the stronger fact holds and was checked: every
connected class on at most three jobs has $\OPT = \OPTLP$, so all $58$
padding classes have $\OPT \le 2$.

\section{Witness counts by machine count}\label{sec:machaxis}

Theorems~\ref{thm:unique} through~\ref{thm:nine} grade the class by job
count and stop where the enumeration does. Grading it by \emph{machine}
count instead makes each level finite, and finishes it.

The reason is Corollary~\ref{cor:cell}. A witness has $\OPTLP = 2$, so by
Lemma~\ref{lem:weight} its total weight is at most $2m$, and since sizes
are at least 1 so is its job count. \emph{Every machine count is therefore
a bounded computation}, with no appeal to how many jobs one is willing to
sweep. The table stops at six machines because larger enumerations are too costly, not because the weight bound fails.

This needs a differently graded sweep. The enumerations behind
Theorems~\ref{thm:unique}--\ref{thm:nine} are indexed by job count, and a
job-count sweep cannot reach twelve jobs at six machines. A
\emph{weight}-graded one can: instances are multiset multiplicity vectors
over the $2(m + \binom{m}{2})$ descriptors, generated level by level in
total weight and reduced to one representative per machine-relabeling
class, with the class counts cross-checked at every level against a
Burnside count graded by weight, an arithmetic sharing no code with the
enumeration. Script: \texttt{p29\_weight\_sweep.py}.

\begin{theorem}[the landscape at three to six machines]\label{thm:machaxis}
Up to machine relabeling and machines no job can use, the witnesses of the
class on at most $m$ machines number $0$, $4$, $111$ and $3{,}293$ for
$m = 3, 4, 5, 6$ respectively, at \emph{any} number of jobs. In
particular there is no three-machine witness, and the four-machine
landscape consists of exactly $\Istar$ and the three seven-job classes of
Table~\ref{tab:seven}.
\end{theorem}

\begin{proof}
By Corollary~\ref{cor:cell} and Lemma~\ref{lem:weight} a witness on $m$
machines has weight at most $2m$, so enumerating every class of weight at
most $2m$ decides the machine count. The four sweeps are
\begin{center}
\begin{tabular}{cccrr}
\hline
machines & weight bound & classes & decided & witnesses \\
\hline
3 & 6  & 567       & 376       & 0 \\
4 & 8  & 9{,}789     & 7{,}289     & 4 \\
5 & 10 & 203{,}027   & 163{,}838   & 111 \\
6 & 12 & 4{,}955{,}391 & 4{,}231{,}347 & 3{,}293 \\
\hline
\end{tabular}
\end{center}
\noindent
with Burnside agreement at every weight level of every range. The counts
are cumulative in support, so each row contains the ones above it.
\end{proof}

Two of the four rows are cross-checked against the job-count sweeps on the
range the two gradings share, and agree as \emph{sets} up to isomorphism
rather than merely in count: at five machines the $59$ classes with at
most eight jobs are exactly the hits of the job-count sweep at $m=5$,
$n \le 8$, split $1/10/48$ at six, seven and eight jobs. The six-machine
sweep reproduces more: it returns exactly one witness at six jobs and
exactly thirteen at seven, which is Theorem~\ref{thm:seven}'s
classification recovered by a different enumerator on a different grading,
and its $669$ nine-job classes of support at most six sit inside
Theorem~\ref{thm:nine}'s $1{,}662$. None of these were designed as
controls.

\paragraph{Decided twice versus certified.}
These are two different guarantees, and only the first is a matter of
agreement between two solvers.

\begin{measurement}[the weight-graded sweeps, decided twice]
\label{meas:machaxis}
\sloppy
The four sweeps of Theorem~\ref{thm:machaxis}---\texttt{p29\_weight\_sweep.py}
at $m = 3,4,5,6$, weight bound $2m$, no minimum-density prune---are decided
twice, in floating
point and again by the exact rational simplex, with identical decision
counts and identical hit sets at three, four, five and six machines
(\texttt{weight\_\allowbreak sweep\_\allowbreak m3\_\allowbreak w6\_\allowbreak rho0.json},
\texttt{weight\_\allowbreak sweep\_\allowbreak m4\_\allowbreak w8\_\allowbreak rho0.json},
\texttt{weight\_\allowbreak sweep\_\allowbreak m5\_\allowbreak w10\_\allowbreak rho0.json},
\texttt{weight\_\allowbreak sweep\_\allowbreak m6\_\allowbreak w12\_\allowbreak rho0.json} and the four
\texttt{\_exact} counterparts; compared as sets by
\texttt{p29\_hitset\_compare.py}). At six
machines that is $4{,}231{,}347$ decisions on each side and the same
$3{,}293$ witnesses, agreeing as \emph{sets} and not merely in number, with
the per-weight breakdown $4$, $19$, $162$, $652$, $2{,}456$ at weights $8$
through $12$ matching term by
term and summing to $3{,}293$. The four sweeps together make $376 + 7{,}289
+ 163{,}838 + 4{,}231{,}347 = 4{,}402{,}850$ decisions. No decision behind
the theorem rests on a floating-point solver.
\end{measurement}

The two halves of that claim were settled at very different times, and they
differ in kind. An exhaustive count is two claims---each
listed instance is a witness, and nothing else is---which fail in opposite
directions and cost differently. The first is a feasibility verdict, and
feasibility is self-certifying: for each of the $3{,}293$ we recomputed
$\OPT$ by exact integer brute force over schedules, using an implementation
independent of the sweep's own, then confirmed by exact rational simplex
that the relaxation is feasible at 2 and infeasible at 1. All $3{,}293$
pass, as do the $4$ and $111$ below them (\texttt{p29\_certify\_hits.py}),
and that half took seconds. The second---that no \emph{further} witness
exists---rests on the other $4{,}228{,}054$ decisions, all infeasibility verdicts, which nothing
certifies and which is the entire reason the exact tier was run at all. It
is the more expensive half by orders of magnitude, and reporting the two
together would conceal that.

\section{The structure at threshold 2}\label{sec:t2}

Every witness this note exhibits has relaxation value 2, and at that
threshold both sides of the gap acquire explicit combinatorial form,
because the configurations are tiny: one weight-2 job, or at most two
unit jobs. We record both forms. They are elementary, assembled from
classical orientation results, and we make no novelty claim for them;
their value here is that together they say exactly \emph{what a
threshold-2 witness is}, turning the sweeps' hits into certificates.

\begin{lemma}[the relaxation at threshold 2]\label{lem:t2frac}
The configuration linear program is feasible at $T = 2$ iff there are
fractional assignments $x_{j,v} \ge 0$, each job's shares supported on
its allowed machines and summing to 1, such that at every machine $v$,
writing $B(v)$ for the total share of weight-2 jobs and $U(v)$ for the
total share of unit jobs at $v$:
\[
2B(v) + U(v) \le 2, \qquad\text{and}\qquad
B(v) + x_{u,v} \le 1 \ \text{ for every unit job } u \text{ allowing } v.
\]
Both inequalities are required at every machine $v$, and the second at
every unit job allowing it.
\end{lemma}

In words: the first inequality is the load bound at threshold 2, weight-2
jobs counting double. The second is the restrictive one. At threshold 2 a
configuration holds either one weight-2 job or at most two unit jobs, never
both, so those two kinds of configuration compete for the same unit of
machine weight; the weight-2 share at $v$ and any \emph{single} unit job's
share at $v$ therefore cannot together exceed 1. The lemma lets one test
threshold-2 feasibility without enumerating a configuration: the displayed
load inequality and family of per-job inequalities suffice.

\begin{proof}
A machine's configurations at $T = 2$ are a single weight-2 job, at most
two unit jobs, or nothing. Covering unit shares $(x_{u,v})_u$ by singles
and pairs costs machine time exactly
$\max\bigl(U(v)/2,\ \max_u x_{u,v}\bigr)$: both are lower bounds (a
configuration covers at most two units' worth of share, and covers any
one job at rate at most its weight $y$), and the bound is attained by
McNaughton's wrap-around rule. Adding time $B(v)$ for the weight-2
configurations, feasibility at $v$ is
$B(v) + \max(U(v)/2, \max_u x_{u,v}) \le 1$, which is exactly the two
displayed families.
\end{proof}

The second family is not implied by the first, so the configuration
linear program is strictly stronger than the assignment relaxation
already at threshold 2: a machine with $B = 1/4$ and one unit job at
share 1 has load $3/2 \le 2$ but needs $5/4$ units of time. Adding
big-job constraints to the assignment relaxation is not new. The
relaxation Wang and Sitters work with \cite{WS16} already limits each
machine to one big job; it is not their addition, and they say so: ``the
following set of linear constraints is equal to (LP3) in [3]'', their [3]
being \cite{EKS14}. In our normalization that constraint is $B(v) \le 1$,
which the display strengthens. Strengthening it instead by interval
constraints, they prove the resulting linear program exact for two-valued
restricted assignment on intervals when the optimum is below twice the big
size; the interval structure excludes $\Istar$, whose support is a
complete graph.

\begin{lemma}[schedules at threshold 2]\label{lem:t2int}
$\OPT \le 2$ iff every machine's singleton load is at most 2 and one
can choose a \emph{servant} for every non-singleton weight-2
job---an endpoint carrying no singleton load, no two such jobs
sharing one---such that the \emph{two-machine} unit jobs admit an
orientation with
in-degree at most the residual capacity: $0$ at servants and at
machines with singleton load 2, otherwise $2$ minus the singleton load.
Singleton jobs of either size are folded into the loads, not oriented:
a one-machine job has no orientation to choose.
For a fixed choice of servants such an orientation exists iff every set
$S$ of machines contains at most $\sum_{v \in S} \mathrm{cap}(v)$ unit
jobs with all allowed machines in $S$ (Hakimi's condition; it is
\cite[Theorem 1]{FG76} with $u = \mathrm{cap}$, and \cite{Hakimi65} for
the directed-degree form), decidable by a single maximum-flow computation.
\end{lemma}

\begin{proof}
A machine at threshold 2 serves either one weight-2 job and nothing
else---so servants take exactly one such job, carry no singleton, and
are distinct---or unit load at most its residual capacity. Singleton
jobs are forced, so they fold into the capacities; a machine with
singleton load 3 or more is infeasible outright, including one no unit
job touches, which is why the load condition is stated separately. The
orientation statement is then Hakimi's theorem in its flow form.
\end{proof}

A threshold-2 witness is therefore an instance satisfying
Lemma~\ref{lem:t2frac} for which \emph{every} servant system fails
Hakimi's condition. On $\Istar$, this is the parity trap of Theorem~\ref{thm:istar} in certificate form: the two diagonals' servants
are one endpoint of each, leaving two \emph{adjacent} cycle vertices
free, and the four unit edges of the cycle cannot orient into two
adjacent vertices at capacity two each; one edge is always left over. Both
characterizations were checked mechanically against the trusted deciders
on every enumerated class of the small sweeps
(Section~\ref{sec:verify}); we state them only at threshold 2, which is
where every known witness of the class lives, and make no claim for
larger thresholds.

A stronger target is a single min--max theorem---``$\OPT \le 2$ iff
every machine set satisfies a checkable inequality''---that would
eliminate Lemma~\ref{lem:t2int}'s existential quantifier over servant
systems. No such theorem exists unless $\mathrm{NP} = \mathrm{coNP}$.

\begin{theorem}[no min--max characterization at threshold 2]\label{thm:conp}
Deciding $\OPT \le 2$ for the class is NP-complete, and remains
NP-complete restricted to instances whose configuration linear program
is feasible at $T = 2$. Consequently, recognizing threshold-2 witnesses
is coNP-complete, and no characterization of the form ``$\OPT \le 2$ if
and only if every member of a family of conditions holds, each indexed
by a polynomial-size object and checkable in polynomial time'' exists
unless $\mathrm{NP} = \mathrm{coNP}$ (Figure~\ref{fig:conp}).
\end{theorem}

\begin{proof}
Membership: a schedule certifies $\OPT \le 2$; for witness recognition's
complement, the configuration linear program at $T = 2$ has polynomially
many columns (a configuration is one weight-2 job or at most two unit
jobs), so a Farkas certificate of infeasibility is polynomial, and a
schedule at $2$ covers the remaining case.

Hardness is a reduction from satisfiability with clause sizes two and
three and every literal occurring at most twice. \cite[Theorem 2.1]{Tovey84}
gives NP-completeness for clause sizes two and three with at most three
occurrences per \emph{variable}---a bound on the variable, not on either of
its literals---and one preprocessing pass sharpens it to the
per-\emph{literal} bound we use. In a formula meeting Tovey's bound, a
literal occurring three times forces its variable to occur three times in
that single polarity, and a variable bound of three then leaves no
occurrence for the opposite literal: the literal is \emph{pure}. Setting a
pure literal true and deleting the clauses it satisfies preserves
satisfiability, leaves every clause size at two or three, and can only lower
occurrence counts, so iterating removes every three-occurrence literal in
polynomial time. Given such a formula: one machine $X_v$
and $\bar X_v$ per variable, one machine $C_c$ per clause; a size-2 job
on $\{X_v, \bar X_v\}$ per variable; a unit job on $\{C_c, L\}$ for each
literal $L$ of clause $c$; and one forced unit job on $C_c$ exactly when
$|c| = 2$, so that $C_c$'s residual capacity is $|c| - 1$.

If the formula is satisfiable, send each variable job to its false
literal's machine; in each clause route one unit edge to a chosen true
literal and the rest into $C_c$. A true literal receives at most its two
occurrences, so no load exceeds $2$. Conversely, in any schedule of
makespan $2$ the machine holding a variable job carries nothing else, so
setting $v$ true when the job is assigned to $\bar X_v$ is well defined; since
$C_c$ receives at most $|c| - 1$ unit edges, some edge is assigned to a literal
machine, which is therefore free of its variable job: that literal is
true and $c$ is satisfied.

Every image is feasible at $T = 2$ fractionally: split each variable job
$\tfrac12$--$\tfrac12$, and each unit edge of clause $c$ as
$(|c|-1)/|c|$ on $C_c$ and $1/|c|$ on the literal; both families of
Lemma~\ref{lem:t2frac} hold, with equality in places. Orienting every
unit edge into its clause machine shows $\OPT \le 3$ unconditionally,
and no weight-2 job fits any configuration below $T = 2$, so
unsatisfiable formulas map exactly to witnesses, at gap exactly $3/2$.

For the corollary, one violated condition from such a family would be a
polynomial-size, polynomially checkable certificate of $\OPT > 2$,
placing the NP-complete decision above in coNP.
\end{proof}

The ratio form of this hardness is due to Asahiro, Jansson, Miyano,
Ono, and Zenmyo \cite{AJMOZ11}---their Theorems 6 and 7 give strong
NP-hardness and rule out even pseudo-polynomial algorithms below $3/2$,
already for orientations with weights exactly $\{1,2\}$---and,
independently, to Ebenlendr, Kr\v{c}\'al, and Sgall \cite{EKS08}, as
\cite{EKS14} recounts. The reduction above also locates the hardness: it survives
restriction to the fractionally feasible regime, where every hard image
is either schedulable at $2$ or a witness. The classification theorems
of this note are therefore exhaustive over a class whose membership
problem is itself intractable. That is why they proceed by bounded
exhaustion rather than by certificate, and why Lemma~\ref{lem:t2int} is
the best-possible \emph{kind} of integral characterization: an NP
certificate, checkable in polynomial time once the servant system is
exhibited, with no matching coNP certificate unless
$\mathrm{NP} = \mathrm{coNP}$.

\paragraph{Sweep method.}
We enumerate instances as multisets of job descriptors, each a pair
(allowed-set, size). The descriptor alphabet is exhaustive by construction:
every nonempty allowed-set of size 1 or 2 over the $m$ machines, crossed
with both sizes, giving $2\,(m + \binom{m}{2})$ descriptors: 20 at four
machines, 30 at five, 42 at six, 12 at three. Job order is irrelevant, so
multisets suffice. Wherever the resulting enumeration is feasible we
deliberately do not quotient out machine relabelings: the enumeration is
then redundant, but the exhaustiveness argument stays elementary, and
redundancy is the safe direction for a claim that nothing exists. The
multiset counts are checkable in closed form: summing
$\binom{d+k-1}{k}$ over $1 \le k \le n$ for $d$ descriptors gives
230{,}229 instances at $(m{=}4, n{\le}6)$; 324{,}631 at $(5, 5)$;
1{,}533{,}938 at $(6, 5)$; and 293{,}929 at $(3, 9)$.

The three six-job ranges above four machines are where the redundant
enumeration first becomes impossible: their raw counts (1{,}947{,}791 at
five machines, 12{,}271{,}511 at six, 61{,}474{,}518 at seven) put it out of
reach. There, and at the seven-, eight- and twelve-job ranges, and at the
$(3, n \le 6)$ range carrying Theorem~\ref{prop:mach}, we instead
enumerate up to machine relabeling, one canonical representative per
isomorphism class, built level by level: every $(k{+}1)$-job class is
reached by extending some $k$-job class by one descriptor, since deleting a
job from any instance gives a class already present one level down.
Because a canonical-form enumerator is new trusted code standing where an
elementary argument used to be, it is never used alone; the two checks
licensing it are described in Section~\ref{sec:verify}.

Each enumerated instance is handled in two steps. First its integral
optimum is computed by brute force. Then it is either pruned or tested at a
single threshold.

A single prune is the only filter. An instance is skipped without any solver
call when $\lfloor 2\,\OPT/3 \rfloor$ falls below its largest job size.
Nothing is lost: the relaxation cannot be feasible below the largest single
size, so such an instance cannot reach gap $3/2$. Comparing the enumerated
totals above with the decision counts below gives the prune's actual effect
per sweep: 16\% at $(4, 6)$, 57\% at $(5, 5)$, 72\% at $(6, 5)$, and
under 1\% at $(3, 9)$, where three machines force large optima.

Surviving instances are tested once, at $T_0 = \lfloor 2\,\OPT/3 \rfloor$.
One threshold suffices because $T_0$ is the unique integer at which
feasibility is equivalent to gap $\ge 3/2$: a consequence of sizes being
integers, which makes $\OPTLP$ an integer, together with the monotonicity of
feasibility in $T$. Scripts and per-$(m,n)$ artifacts accompany the note.

\section{Verification}\label{sec:verify}

\paragraph{The need for exact arithmetic.}
A stored primal solution certifies a FEASIBLE verdict beyond argument---any
independent script can check it, as Section~\ref{sec:known}'s exact certification does---but it cannot certify
INFEASIBLE, and infeasibility is what every exhaustive negative above rests
on. A false infeasible from a floating-point solver would silently void the
corresponding claim. Figure~\ref{fig:tiers} summarizes the two tiers.
We therefore re-decided every feasibility question
behind Theorems~\ref{thm:min} and~\ref{thm:unique} and
Theorem~\ref{prop:mach} with no floating point anywhere.
\texttt{p29\_exact\_lp\_feasibility.py} implements the test as a phase-1
primal simplex over Python's exact rationals, with Bland's rule throughout
so every pivot is exact and termination is guaranteed without perturbation.
Feasibility needs no second phase: the instance is feasible at $T$ exactly
when the phase-1 minimum of the artificial variables is zero.

\paragraph{Scope of computational independence.}
The simplex is an independent implementation, not a wrapper around the
floating-point solver. Its \emph{algorithm} is correct by exact arithmetic
and Bland's rule; whether this \emph{code} builds its tableau, initializes
its artificials, and reads off its optimum correctly is a separate question,
and the only evidence bearing on it is comparison against an independent
system. We therefore ran that comparison, against GLOP---the
floating-point simplex shipped with Google OR-Tools---and report it as
Measurement~\ref{meas:crosscheck}(a) rather than leaving it implicit.

\begin{measurement}[the two checks on shared and unshared code]
\label{meas:crosscheck}
\sloppy
(a) \emph{The solvers, against each other.}
\texttt{p29\_\allowbreak exact\_\allowbreak vs\_\allowbreak glop\_\allowbreak crosscheck.py}, run at $200$ random instances
from the swept class ($m$ in $[3,6]$, $n$ in $[1,8]$, degree at most $2$,
sizes in $\{1,2\}$) at seed $20260731$, decides each at every relevant
threshold: $1{,}226$ decisions, $0$ disagreements
(\texttt{p29\_\allowbreak exact\_\allowbreak vs\_\allowbreak glop\_\allowbreak crosscheck.json}).
(b) \emph{The configuration enumeration, against a naive reference.}
\texttt{p29\_\allowbreak enumerate\_\allowbreak configs\_\allowbreak completeness\_\allowbreak check.py}, run at $5{,}000$
random cases at seed $20260731$, compares the enumeration shared by both
pipelines against an unpruned all-subsets reference as \emph{sets}: $0$
mismatches (\texttt{p29\_\allowbreak enumerate\_\allowbreak configs\_\allowbreak completeness\_\allowbreak check.json}).
The two are not interchangeable. (a) can say nothing about any stage the
two pipelines share, which is why (b) exists; (b) tests one shared stage
only, and the paragraph below says which stages remain untested by either.
\end{measurement}

The independence is partial, and we state the boundary rather than leave
the reader to find it. Both pipelines call the same configuration enumeration,
the same brute-force routine for the integral optimum, the same pruning
rule, and the same threshold selection; only the arithmetic and the solve
differ. Two consequences follow. First, matching decision \emph{counts}
across the two runs is not evidence about the solvers at all---the counts
are fixed by the shared enumerator---so we do not offer it as such; what
counts as evidence is that the two runs return the same \emph{verdicts}, and at
$n = 6$ the same hit set. Second, the shared stages are exactly the ones this comparison cannot
police, and neither of them fails in a safe direction. A missing
configuration column can only remove a way to cover a job, biasing toward
INFEASIBLE; and since a hit requires the relaxation to be \emph{feasible} at
the threshold, that bias manufactures precisely the negatives an exhaustive
claim consists of. A wrong integral optimum is worse, because it sets the
threshold: $T_0 = \lfloor 2\,\OPT/3 \rfloor$ comes from the shared
brute-force routine, so an understated optimum either lowers $T_0$ or trips
the prune above, and in both cases turns a witness into a negative in
silence. We therefore
pin the enumeration against a naive all-subsets reference, as
Measurement~\ref{meas:crosscheck}(b) records, and validate the remaining
shared stages end to end by the positive controls below. The integral optimum is where this section's guarantees stop: its computation has no directional guarantee. Three things constrain it, none
of them a proof. The pipeline is required to rediscover $107$ known witness
classes---one at six jobs, nine at seven, ninety-seven at eight---before any
negative result of the same run is believed, and each forces this routine to
return $3$. Every certified hit had its optimum recomputed by an
implementation independent of the sweep's own, as recorded above. And the
independent enumerator described below recomputes it with a routine of its
own at four, five and six machines, agreeing at each. No independent optimum covers the enumeration \emph{as a whole}---the infeasibility
verdicts, where the routine's failure direction is silent---and supplying
one is the cheapest strengthening this verification still lacks. The remark
after Theorem~\ref{thm:unique} closes that theorem's six- and seven-machine
ranges without a sweep, so the exposure is removed there; it is not removed
from Theorem~\ref{thm:min}, through which that theorem reaches them.

\paragraph{The re-verification.}
Sweep by sweep, with the artifact each count is read from, and with the
overlaps between ranges counted once rather than added.

\begin{measurement}[every job-count sweep, rerun in exact arithmetic]
\label{meas:reverify}
\sloppy
Every sweep behind Theorems~\ref{thm:min} and~\ref{thm:unique} and
Theorem~\ref{prop:mach} was rerun with no floating point anywhere
(\texttt{p29\_\allowbreak exact\_\allowbreak minimality\_\allowbreak sweep.py} for the redundant sweeps,
\texttt{p29\_canonical\_sweep.py --exact} for the canonical ones, both over
\texttt{p29\_\allowbreak exact\_\allowbreak lp\_\allowbreak feasibility.py}); the floating-point tier is
\texttt{minimality\_sweep.py} and \texttt{p29\_\allowbreak canonical\_\allowbreak sweep.py}. Decision
counts, artifact by artifact:
$192{,}847 + 139{,}975 + 433{,}743 + 292{,}971 = 1{,}059{,}536$ solver calls
for the redundant sweeps, at $(m, n \le)$ equal to $(4,6)$, $(5,5)$, $(6,5)$
and $(3,9)$ (\texttt{minimality\_\allowbreak sweep\_\allowbreak m4\_\allowbreak n6.json},
\texttt{minimality\_\allowbreak sweep\_\allowbreak m5\_\allowbreak n5.json},
\texttt{minimality\_\allowbreak sweep\_\allowbreak m6\_\allowbreak n5.json},
\texttt{minimality\_\allowbreak sweep\_\allowbreak m3\_\allowbreak n9.json}).
The first term is the four-machine sweep extended to six jobs, which
subsumes the four-machine five-job sweep (34{,}200,
\texttt{minimality\_\allowbreak sweep\_\allowbreak m4\_\allowbreak n5.json}) rather than adding to
it; we count it once. These are calls, not distinct instances, since the
ranges overlap: a four-machine instance reappears in the six-machine
enumeration with two unused machines. The canonical sweeps were likewise
rerun in exact arithmetic: the four six-job ranges (four through seven
machines, $61{,}654$ calls), the five seven-job ranges behind
Theorem~\ref{thm:seven} (four through eight machines,
$39{,}500 + 90{,}796 + 130{,}508 + 145{,}339 + 149{,}073 = 555{,}216$
calls, which subsume the six-job ranges instance-for-instance;
\texttt{canonical\_\allowbreak sweep\_\allowbreak m4\_\allowbreak n7.json} through
\texttt{canonical\_\allowbreak sweep\_\allowbreak m8\_\allowbreak n7.json}), and the
three-machine sweep to twelve jobs behind Theorem~\ref{prop:mach}
(459{,}698 calls over $459{,}889$ classes,
\texttt{canonical\_\allowbreak sweep\_\allowbreak m3\_\allowbreak n12.json}, which subsume that theorem's own
six-job range, $3{,}237$ calls,
\texttt{canonical\_\allowbreak sweep\_\allowbreak m3\_\allowbreak n6.json}, instance-for-instance), and the six
eight-job ranges behind
Theorem~\ref{thm:eight} (four through nine machines, $138{,}950 +
446{,}708 + 846{,}081 + 1{,}111{,}144 + 1{,}208{,}344 + 1{,}233{,}378 =
4{,}984{,}605$ calls, \texttt{canonical\_\allowbreak sweep\_\allowbreak m4\_\allowbreak n8.json} through
\texttt{canonical\_\allowbreak sweep\_\allowbreak m9\_\allowbreak n8.json}, subsuming the seven-job ranges
instance-for-instance at four through eight machines):
$1{,}121{,}190 + 555{,}216 + 459{,}698 + 4{,}984{,}605 = 7{,}120{,}709$
in all, the first summand being the redundant total above plus the
$61{,}654$ six-job canonical calls. The nine-job layer of
Theorem~\ref{thm:nine} ($4{,}231{,}656$ decisions over the five nonempty
support shards of \texttt{p29\_\allowbreak connected\_\allowbreak sweep.py}, supports four through
eight; \texttt{connected\_\allowbreak sweep\_\allowbreak s4\_\allowbreak n9.json} through
\texttt{connected\_\allowbreak sweep\_\allowbreak s8\_\allowbreak n9.json}) and the weight-graded layer of
Theorem~\ref{thm:machaxis} are counted separately, and added in the Code
and data availability section.
\end{measurement}

Every verdict agreed. On the all-negative sweeps both procedures returned
no hits; on the six-job four-machine sweep both returned the same three
instances, listed in Theorem~\ref{thm:unique}; on every canonical sweep
both returned the same hit set: one class per six-job range; $4$, $11$,
$14$, $14$, $14$ across the seven-job ranges at four through eight
machines; and up to 168 per eight-job range. Agreement on hit
\emph{sets}, not merely on counts, is the claim: a disagreement in either
direction would surface as an instance present in one artifact and absent
from the other.

\paragraph{Validating the canonical enumerator.}
The enumeration up to relabeling behind Theorem~\ref{thm:unique}'s five-,
six-, and seven-machine ranges is new trusted code standing where an
elementary argument used to be: if it dropped an isomorphism class, the
sweep would silently miss instances and its negatives would be worthless,
the same failure mode the positive control exists to catch, in a new
place. Three checks license it. First, on every range it runs, its class
counts are compared against Burnside's lemma over the machine-relabeling
group---an independent arithmetic that shares no code with the enumeration
---and match on every job count of every range. Second, on small ranges
(up to four machines, three jobs) a third count, brute-force
canonicalization of every raw multiset, agrees with both. Third, at four
machines and six jobs, where the redundant sweep is feasible, the two hit
sets are compared up to isomorphism: the redundant sweep's three raw hits
collapse to exactly the canonical sweep's single class, with raw
multiplicity 3 matching the orbit size computed in
Theorem~\ref{thm:unique}. During development the first two checks each
caught a real defect (an enumerator that failed to canonicalize, and a
miscount in the Burnside routine: errors that were also mutually
inconsistent, and inconsistent with a hand count); both were fixed before
the enumerator was pointed at any range whose answer was unknown, and the
project log records the episode.

One limitation of the first check should be stated rather than left for a
reader to find. A Burnside match constrains the class \emph{counts}; it
does not by itself imply that the enumerator's hit detection is complete,
because a canonicalizer that split one class in two while dropping another
would leave the count invariant. The negatives rest on redundancy of a different kind. Each class list is
produced separately by five ranges at seven jobs and six at eight, which
catches a fault confined to one range though not a fault in the shared
canonicalizer; the weight-graded sweep of Section~\ref{sec:machaxis}, a
separately coded enumerator, reproduces the same classes as sets on the
range the two share, as a check of the deposited result files confirms; and the pipeline is
required to rediscover $\Istar$, and at seven jobs the
Jansen--Land--Maack instance, which it was not designed around.

A later adversarial review requested a second overlap point: an enumerator
written separately
that uses \emph{no} canonical forms and no isomorphism reduction at all,
scanning plain multisets of job types and deciding each with its own
brute-force optimum and its own exact rational certificates over a common
denominator; no simplex anywhere, and no code shared with the harness.

\begin{measurement}[the six-job theorem, recomputed by unshared code]
\label{meas:independent}
\sloppy
\texttt{p29\_\allowbreak redundant\_\allowbreak independent.py}, run at six
jobs and four, five and six machines. At
four machines it scans $177{,}100$ instances and returns exactly the three
raw witnesses of $\Istar$'s orbit, with none outside it
(\texttt{p29\_\allowbreak redundant\_\allowbreak independent.json}); at five it scans
$1{,}623{,}160$ and
returns the fifteen of the five-machine orbit, with nothing outside it
(\texttt{p29\_\allowbreak redundant\_\allowbreak independent\_\allowbreak m5.json}); at
six it scans all $\binom{47}{6} = 10{,}737{,}573$ and returns the
$\binom{6}{4}\cdot 3 = 45$ of the six-machine orbit, again with nothing
outside it (\texttt{p29\_\allowbreak redundant\_\allowbreak independent\_\allowbreak m6.json}). Each artifact
also records how many of its instances passed the weight bound before any
optimum was computed: $64{,}350$, $1{,}409{,}980$ and $10{,}737{,}573$
respectively. In all three runs the witnesses found equal the orbit size and
the count of orbit members missed is $0$.
\end{measurement}

\noindent
The redundant scan is a genuinely independent second computation of the
six-job uniqueness theorem at four, five and six machines, and it agrees at
each. Seven machines it does not reach; that range is the one the remark
after Theorem~\ref{thm:unique} closes without a sweep.

\paragraph{Checking the threshold-2 characterizations.}
These structural claims were checked exhaustively where the sweeps reach
and adversarially where they do not.

\begin{measurement}[Lemmas~\ref{lem:t2frac} and~\ref{lem:t2int} against the
trusted deciders]\label{meas:t2}
\sloppy
On every canonical class of the four- and five-machine six-job ranges,
both characterizations were compared against the trusted deciders---the
configuration solver for the fractional side, the brute-force integral
optimum for the schedules---by
\texttt{p29\_\allowbreak threshold2\_\allowbreak structure.py} run as \texttt{validate 4 6} and
\texttt{validate 5 6}: $11{,}556 + 22{,}951 = 34{,}507$ classes, and on
each side zero mismatches, fractional and integral alike
(\texttt{threshold2\_\allowbreak validation\_\allowbreak m4\_\allowbreak n6.json},
\texttt{threshold2\_\allowbreak validation\_\allowbreak m5\_\allowbreak n6.json}). They were then
independently stress-tested on $1{,}400$ adversarial instances with
parallel multiplicities and singleton loads beyond the sweeps' range,
by \texttt{p29\_\allowbreak threshold2\_\allowbreak referee\_\allowbreak check.py}
run at seed $1$ for $1{,}400$ trials; that run prints its verdict and
deposits no artifact of its own, so the count above is the invocation's
parameter rather than a figure read back from a file.
\end{measurement}

\paragraph{Checking the rounding argument.}
The analytic half of Theorem~\ref{prop:mach} is a proof, not a
computation, but its two essential steps were still checked mechanically.

\begin{measurement}[the two steps of the rounding argument]
\label{meas:rounding}
\sloppy
\texttt{p29\_\allowbreak three\_\allowbreak machine\_\allowbreak rounding\_\allowbreak check.py}
checks the
subset-sum gap claim, exhaustively over pools of up to eight unit and
eight weight-2 jobs; and the bound $\OPT \le T + 2$, against $512{,}000$
exhaustive small instances---singleton loads in $\{0,\dots,4\}$ on each of
three machines and each of three pools carrying $\{0,\dots,3\}$ unit and
$\{0,\dots,3\}$ weight-2 jobs, so $5^3 \cdot 4^6 = 512{,}000$ in
all---and $20{,}000$ random large ones at seed $7$, tested at
the assignment relaxation's threshold, which never exceeds $\OPTLP$, a
strictly harder test than the theorem needs. Zero violations; the
observed maximum of $\OPT$ minus the fractional threshold was $1$. This
script prints its verdict and deposits no artifact of its own, so the
counts above are its enumeration bounds rather than figures read back from
a file.
\end{measurement}

\noindent
The argument was also reviewed adversarially by ChatGPT and Gemini before
entering this note (see the Use of AI tools section).

\paragraph{Positive control.}
An exhaustive negative from a pipeline never observed to return a positive
is worth little: a fault in enumeration, pruning, threshold selection, or
hit detection would produce exactly that evidence, and those stages are
shared between the two solvers. We therefore required the pipeline to
rediscover $\Istar$. The four-machine sweep extended to six jobs reproduces
the smaller counts identically and returns three hits at gap exactly $3/2$
(Theorem~\ref{thm:unique}), in both the floating-point and the exact run.
Every stage---enumeration, pruning, threshold selection, the solve, and hit
detection---is thus exercised on a known positive, which is what licenses
reading the other sweeps' silence as absence rather than as blindness. The
canonical sweeps carry the same control internally: every range contains
$\Istar$'s class, and every range returned it. The seven-job sweeps add a
control of a different kind: they must---and do---rediscover the
\cite{JLM16} instance, a published witness the pipeline was never designed
around, identified among the hits by explicit relabeling
(Theorem~\ref{thm:seven}).

\paragraph{Exact-verification cost and a cheaper design.}
Deciding everything twice is what bounds how far this method reaches. On
the eight-job lane the floating-point pass took about $2.6$ hours on one
core and the exact re-decision about $10.3$, a factor of four; at nine jobs
that factor is applied again, to a layer of comparable size rather than a
larger one: the connected nine-job sweep makes $4{,}231{,}656$ decisions
against the eight-job lane's $4{,}984{,}605$, because it enumerates each
connected class once at its exact support instead of re-deciding it in every
range above. The cost
is avoidable in principle, and the reason is
the asymmetry this section opened with. A FEASIBLE verdict is certified by
its own primal solution, so the exact tier does no work there that a stored
witness could not do more cheaply. An INFEASIBLE verdict is what needs the
exact solve. And it, too, has a short certificate. Writing the relaxation
at threshold $T$ as $\{y \ge 0 : Ay = b\}$, Farkas' lemma says the system is
infeasible exactly when some vector $u$ of multipliers, one per constraint,
satisfies $u^{\mathsf T}A \le 0$ componentwise and $u^{\mathsf T}b > 0$.
Such a $u$ can be read off the floating-point solve's dual ray, reconstructed
over the rationals, and checked by a few exact inner products, work linear
in the number of nonzeros, against a full exact simplex per decision. Where
reconstruction fails the exact simplex remains as a fallback, so correctness
would not depend on the fast path succeeding. We have not implemented this:
every exact decision reported in this note comes from the simplex tier
described above, and we record the design because it is what we would build
first to push the enumeration further.

\paragraph{Class limits.}
Sizes outside $\{1,2\}$ and jobs allowed on three or more machines lie
outside the exhaustive classification; Section~\ref{sec:search} samples them,
without exhaustive coverage. Enlarging an allowed set relaxes both problems at once---it
can lower $\OPTLP$ but can equally lower $\OPT$---so the per-instance ratio
can move either way. Enlarging $M(v_1v_3)$ in $\Istar$ to all four machines
drops $\OPT$ to 2 and the gap to 1. Only the class-level containment holds:
the supremum over a superclass is at least the supremum over a subclass. We
note that Huang and Ott's algorithm \cite{HO16} extends to the case where
the smaller jobs may go to any number of machines, which bears on the
degree-3 direction.

\section{Searching above 3/2}\label{sec:search}

The class studied above is closed at $3/2$ by \cite[Corollary~11]{JLM16}
(Section~\ref{sec:known}), so a
search for a larger gap must leave it. We report such a search for completeness
and as a check on the pipeline. A candidate is scored only when an exact
integer-programming solver proves its optimum; its relaxation bound is
evaluated in floating point, with no exact primal or dual certificate, so
each scored gap is a numerical value rather than a proven one. Each new best gap
within a run is re-evaluated from scratch by the same pipeline before it
is recorded. Three outcomes are separated:
\emph{optimum proven} (the integral optimum proven, the relaxation bound
numerically evaluated; the artifacts' \texttt{certified\_evaluations}),
\emph{degenerate} (a job with no allowed machine, rejected before either
solver runs), and \emph{unproven} (the integer solver returned without
proof inside its limit).

\begin{measurement}[the search above $3/2$, six runs]\label{meas:search}
\sloppy
Six runs, each deterministic in its seed and each depositing its own
artifact, propose $51{,}812$ instances in total: a $4{,}000$-instance
two-size sweep at seed $291$ (\texttt{p29\_two\_size\_sweep.py},
\texttt{gap\_\allowbreak sweep\_\allowbreak two\_\allowbreak sizes.json}); a restarted local search at
$n = 7$, $m = 4$, maximum size $4$, $12$ restarts of $150$ steps, seed
$29$, making $1{,}812$ evaluations (\texttt{gap\_search.py},
\texttt{gap\_search\_n7\_m4.json}); and four seeded searches of
$2{,}000$, $12{,}000$, $12{,}000$ and $20{,}000$ steps at seeds
$2929$, $101$, $202$ and $303$ (\texttt{seeded\_search.py}, artifacts
\texttt{seeded\_\allowbreak search\_\allowbreak result\_\allowbreak s2929.json},
\texttt{seeded\_\allowbreak search\_\allowbreak result\_\allowbreak s101.json},
\texttt{seeded\_\allowbreak search\_\allowbreak result\_\allowbreak s202.json} and
\texttt{seeded\_\allowbreak search\_\allowbreak result\_\allowbreak s303.json}). Summing the six artifacts:
$51{,}180$ with the optimum proven, $632$ degenerate, $0$ unproven, and
the first two sum to the proposed total exactly. The largest gap any run
reports is $3/2$; no numerically evaluated gap exceeded it.
\end{measurement}

\noindent
Of these, the $4{,}000$-instance sweep restricted to degree-2 instances with
sizes $\{1,2\}$ lies inside that closed class, so its
negative result is a consistency check against a known bound rather than
evidence about an open question; the remaining runs, which admit jobs on
more than two machines, do explore genuinely open territory. For graph
balancing the relaxation's gap is at most 1.749 by non-constructive local
search \cite{JR19}.

We adopt a standing rule for future runs, stated here rather than merely
practised: the floating-point pipeline nominates candidates, and only the
exact tier confirms them. No record counts until its optimum is proven, its
bound re-decided in exact rational arithmetic, and its fractional solution
certified.

\section{Discussion}\label{sec:disc}

An explicit witness is easier to build on than a bound inherited from a
hardness reduction, and a minimum, unique one is easier still.
By Theorem~\ref{thm:unique}, a construction seeking a larger gap in a
neighbouring class has only one six-job core of this class to start from, on
any number of machines, and that core already attains the largest gap the
class allows. Because
Theorem~\ref{thm:min} certifies a floor of six jobs, an exhaustive search
for a gap-$3/2$ instance \emph{of this class} may begin at six jobs
rather than one, which is where such searches become tractable. The floor
does not transfer: it is proved within two-weight graph balancing, and the
five-job instance of \cite{VW14} recorded below has gap at least $3/2$ in
the unrelated model.

The seven-job landscape (Theorem~\ref{thm:seven}) shows how quickly
uniqueness fails above the floor: one witness at six jobs, thirteen at
seven. The connected ones are built from a small set of recurring parts
---a size-2 job whose endpoints are pinned at load 2 by interchangeable
``arms'': a singleton plus a doubled pair (the \cite{JLM16} form), a
tripled pair, or a unit triangle closed by a second size-2 job---together
with $\Istar$-like cycle structures, one of which is exactly $\Istar$ with
a matching edge subdivided through a new machine. We state this as observed structure rather than as a theorem. The
pattern does persist at eight jobs, where Theorem~\ref{thm:eight} counts
154 witnesses and no connected one needs more than seven machines.

$\Istar$'s form---a unit cycle plus a weight-2 perfect matching on
$K_4$---generalizes in several directions at once: larger cycles with
chorded matchings, odd structures, and instances with three sizes. The subdivided witness of
Table~\ref{tab:seven} is a first confirmed instance of that family.

\paragraph{Search protocol behind the priority claims.}
The claims that no smaller explicit witness and no prior minimality result
appear in print rest on the following. The primary-source review covered
\cite{EKS14,Svensson11,JR17,JR19,VW14,JLM16,CS16} in full; searched for
explicit gap instances in the graph-balancing and restricted-assignment
literature, including the doctoral theses of the Kiel group; and checked
the two-weight approximation line \cite{AJMOZ11,KM13,HO16,PSO16,CS16},
whose first entry is the earliest: \cite{AJMOZ11} already give a
$3/2$-approximation for orienting simple graphs with weights $\{1,2\}$,
matching their own hardness bound. The review examined every one of
\cite{EKS14,Svensson11,JR17,JR19,VW14,JLM16,CS16,WS16,AJMOZ11,KM13,HO16,PSO16}
in the original. One qualification belongs here rather than in a footnote:
\cite{KM13} is not only an approximation-ratio paper; its Theorem~3.2 states a
configuration-linear-program integrality-gap bound ($5/3 + s$ for sizes
$\{s, 1\}$, without a degree restriction), attributed there to
\cite{Svensson11} and mildly generalized. In our normalization
($s = 1/2$) that bound is vacuous, so nothing in this note rests on it;
but it would be wrong to describe that line of work as bearing only on
approximation ratios.
None of the works surveyed exhibits a configuration-linear-program gap
instance smaller than \cite{JLM16}'s seven jobs, and none states a
minimality or uniqueness result. One near miss deserves naming, since it
is the first thing a reader of that literature will reach for:
\cite{VW14}'s Proposition~2 family at $k = 2$ is a five-job, four-machine
instance of gap at least $3/2$ in which every job has exactly two allowed
machines.
It is outside our class only because the two machines run the job at
different speeds---the model is unrelated graph balancing, not restricted
assignment---and the minimality theorem does not contradict it.
Searches were performed in July--August 2026.

\paragraph{Beyond scheduling: the prior art the survey misses.}
It stays inside scheduling, and that is its limitation. The questions asked
here have been asked before, of the same relaxation, in another field.

Kartak, Ripatti, Scheithauer and Kurz \cite{KRSK15} study the
one-dimensional cutting stock problem in the pattern formulation of Gilmore
and Gomory---the configuration linear program of that problem---where a gap
instance is one failing the integer round-up property, that is, one whose
gap $\Delta(E) = z_D(E) - z_C(E)$ is at least $1$. Fixing the demand $n$,
they decompose the infinite set of instances into finitely many equivalence
classes, enumerate the classes exhaustively, and prove by that enumeration
that every instance with $n \le 9$ has the \emph{proper} integer round-up property, exhibiting classes
without it at $n = 10$. Their Table~1 goes further: for each $n \le 11$ it
gives the number of classes, the maximum gap, the number of classes
attaining it---$365$ at $n = 10$, six at $n = 11$---and the instances
themselves. Their result is a minimality theorem, a count and a classification for
gap instances of a configuration linear program, by exhaustive enumeration
over equivalence classes, published a decade before this note.
\cite{RK20} continue the line, minimizing in the stock length rather than
the demand. The wider genre---exact integrality gaps of small instances by
enumeration modulo isomorphism---is older still and active: \cite{BB08} for
the subtour relaxation of the travelling salesman problem, \cite{ABE06} for
the two-edge-connected subgraph problem, \cite{SBG26} for the asymmetric
travelling salesman problem.

Our results are not the first of their kind; an earlier draft incorrectly claimed priority. The true claim is narrower, and the survey above establishes it: we know of no minimality, uniqueness or
classification result for gap instances of the configuration linear program
of a \emph{scheduling} problem. This note differs from \cite{KRSK15} in two
ways, neither of them a priority claim. Their gap is additive and the
question is whether any gap at all appears; ours is the ratio
$\OPT/\OPTLP$, and the question is about a particular value, $3/2$, already
known to be the exact supremum for the class (Section~\ref{sec:known}).
And where a gap is a supremum approached but never attained---as in the
travelling-salesman line---``the smallest instance attaining it'' is
vacuous, and only a table of gaps by size is possible. In the present class it is attained
at finite size, which is what makes minimality and uniqueness exact
questions rather than limits.

Figure~\ref{fig:future} maps where the work goes next: the counting is
exact through nine jobs and our enumeration does not reach beyond, while the
open problem, where the restricted-assignment gap actually lies in
$[3/2, 11/6]$, is untouched by anything here. That figure and the rest of the
supporting pictures are collected in Appendix~\ref{app:figs}.

\section*{Code and data availability}

The code and artifacts needed to reproduce every computational result reported
in this note are deposited as supplementary material on the Open Science Framework,
\url{https://doi.org/10.17605/OSF.IO/SKX86}, as
\path{osf_package_p29witness.zip}: thirty-one scripts in \path{code/} and six
independent referee probes with their shared library, the result artifacts they write, the associated run
logs, a driver that regenerates all of it, and pinned versions of the principal
dependencies (SciPy, used only by an independent cross-check, is deliberately
left unpinned). The same project also holds the
separate bundles of the two companion papers on the Wang--Sitters rounding,
\path{osf_p29sup.zip} and the five \path{osf_p29lever_*.zip}; the counts here
describe the witness deposit only. The deposit is the replication set, not
the project's working directory: drafts, review correspondence, and
publishers' copies of cited papers are not included. A README maps each claim
below to the script and artifact that establish it.

Three distinct kinds of evidence are provided and should not be conflated.
Exact rational certificates for the two headline instances
(\texttt{p29\_exact\_rational\_verify.py}) use no solver and no floating
point. The minimality, uniqueness, and classification sweeps are decided
twice, once in floating point and once by the exact simplex
(\texttt{p29\_exact\_minimality\_sweep.py} for the redundant sweeps,
\texttt{p29\_canonical\_sweep.py --exact} for the canonical ones), agreeing
on all 7{,}120{,}709 decisions. That total covers the sweeps behind
Theorems~\ref{thm:min} through~\ref{thm:eight}; the nine-job layer's
$4{,}231{,}656$ decisions were double-decided separately
(Theorem~\ref{thm:nine}), and the weight-graded sweeps behind
Theorem~\ref{thm:machaxis} at three, four, five and six machines add
$376 + 7{,}289 + 163{,}838 + 4{,}231{,}347 = 4{,}402{,}850$, for
$15{,}755{,}215$ in all. That last summand is the six-machine sweep, which
was floating point only until its exact tier finished; every exhaustive
range in the note is now decided twice. The exploratory search of Section~\ref{sec:search} computes its
relaxation bounds in floating point, though an exact integer solver proves
each optimum, and it is reported as evidence rather than proof.

Three further scripts matter for checking the argument rather than the
numbers. The embedding search separating $\Istar$ from the instance of
\cite{JLM16} is\\
\texttt{p29\_\allowbreak witness\_\allowbreak distinctness\_\allowbreak check.py}; the one component both
sweep pipelines share is pinned by\\
\texttt{p29\_\allowbreak enumerate\_\allowbreak configs\_\allowbreak completeness\_\allowbreak check.py}; the hit-set
comparison licensing the canonical enumerator is
\texttt{p29\_hitset\_compare.py}. Every artifact records the exact
invocation that produced it, and the deposited README states plainly which
sweeps are expected to return hits---the six-job uniqueness runs, which
double as positive controls---so that a replicator can tell a correct
result from a broken one.

\appendix

\section{Supplementary figures}\label{app:figs}

These figures are explanatory: every proof stands
without them. Four of them show the nine-job layer, which is
Theorem~\ref{thm:nine}; its two halves are established differently, and
the figures mark the distinction rather than hiding it. The connected
count is an exhaustive sweep, every decision made twice and the $978$
hits decided a third time by the independent pipeline of
Section~\ref{sec:verify}. The disconnected count is the composite
calculus, whose classes were constructed explicitly and re-decided on the
assembled instances.
Each figure is generated programmatically from the result artifacts
deposited in the supplementary material, so the pictures are drawn from
the same artifacts as the numbers; the deposit itself contains only the replication set---code,
artifacts, and logs---not the figures or their generator.

\begin{figure}[htbp]\centering
\includegraphics[width=0.95\textwidth]{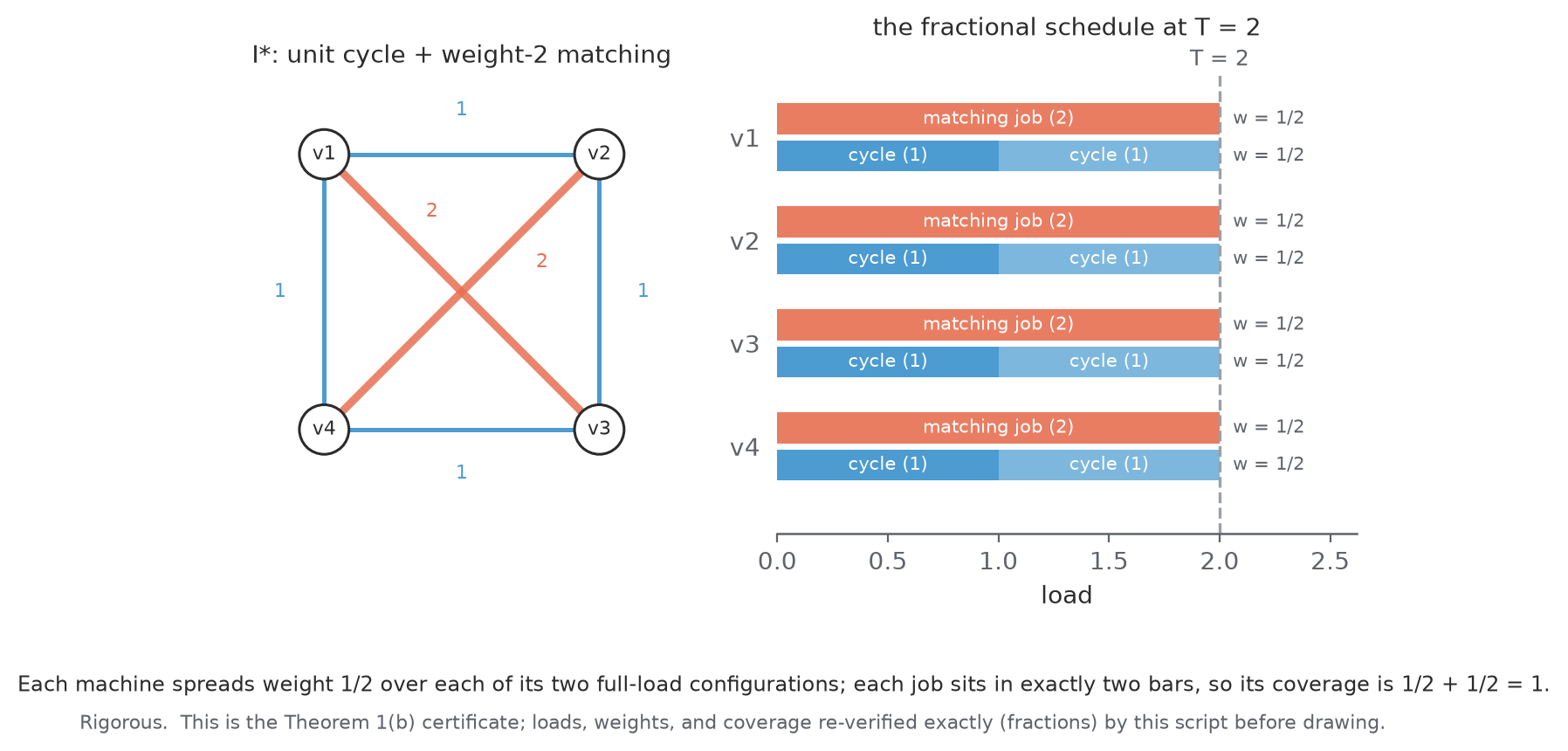}
\caption{The witness $\Istar$ (left) and the fractional schedule reaching
$T=2$: each machine splits weight $1/2$ over its two full-load
configurations, and each job is covered exactly once.}
\label{fig:frac}
\end{figure}

\begin{figure}[htbp]\centering
\includegraphics[width=0.95\textwidth]{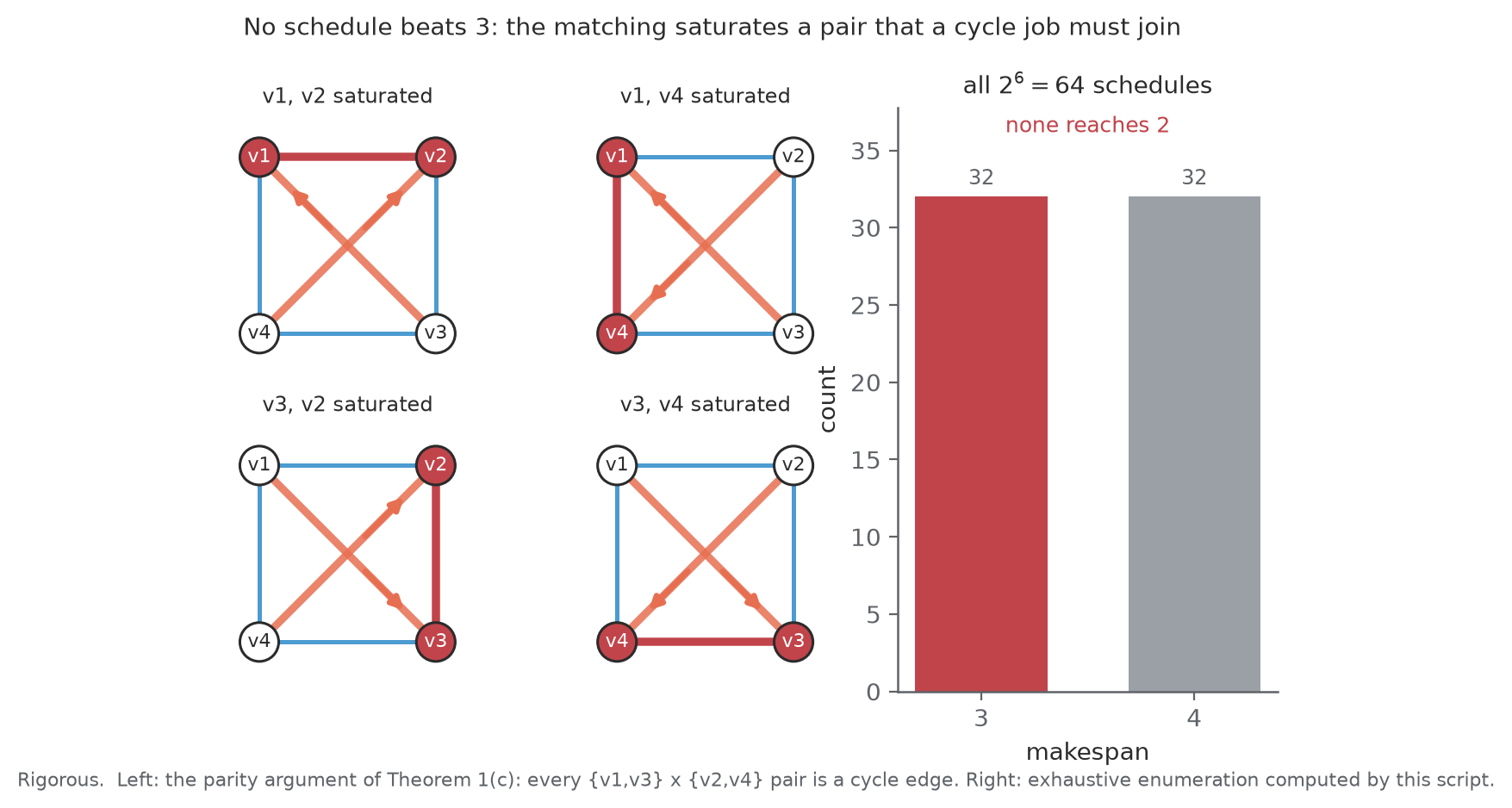}
\caption{No schedule finishes below 3: each orientation of the matching
saturates a pair joined by a cycle job. Right: the exhaustive makespan
distribution over all 64 assignments. None reaches 2.}
\label{fig:trap}
\end{figure}

\begin{figure}[htbp]\centering
\includegraphics[width=0.95\textwidth]{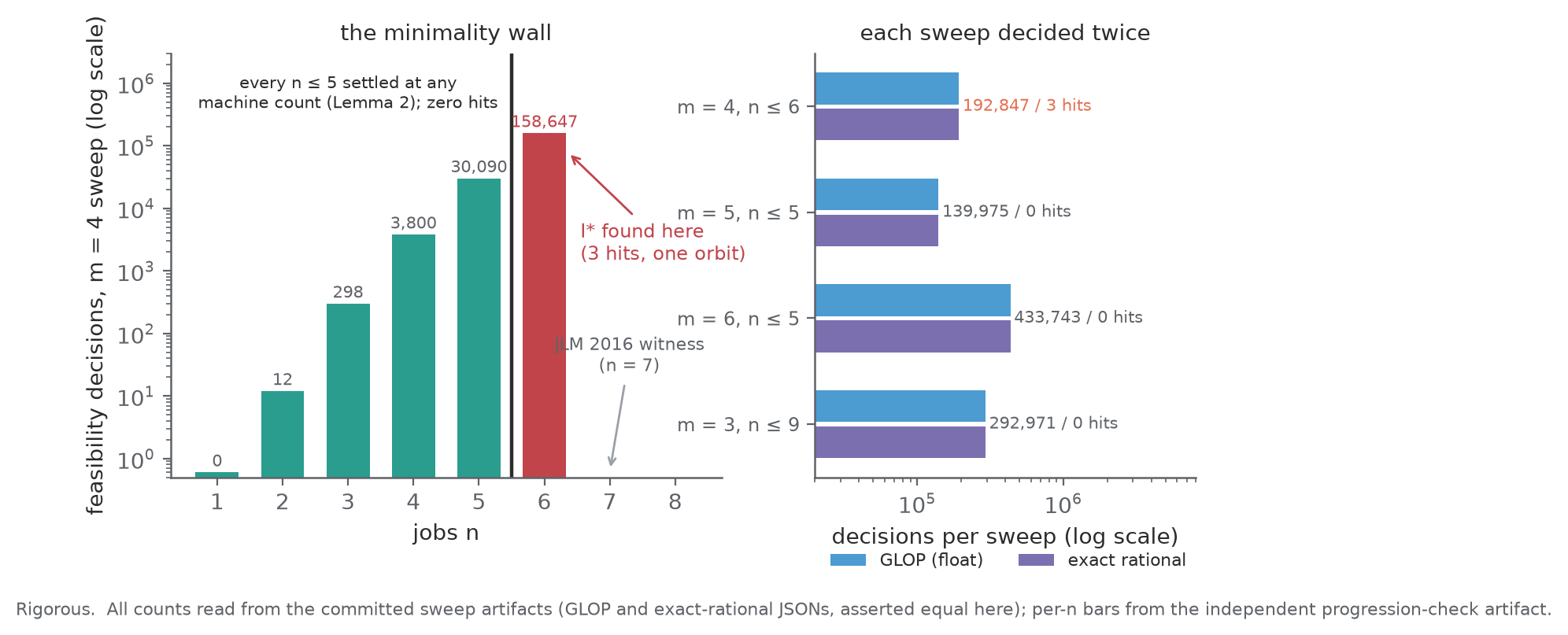}
\caption{The minimality wall: per-$n$ decisions at four machines, and
every sweep decided twice. The six-job bar is the only one with hits.}
\label{fig:wall}
\end{figure}

\begin{figure}[htbp]\centering
\includegraphics[width=0.95\textwidth]{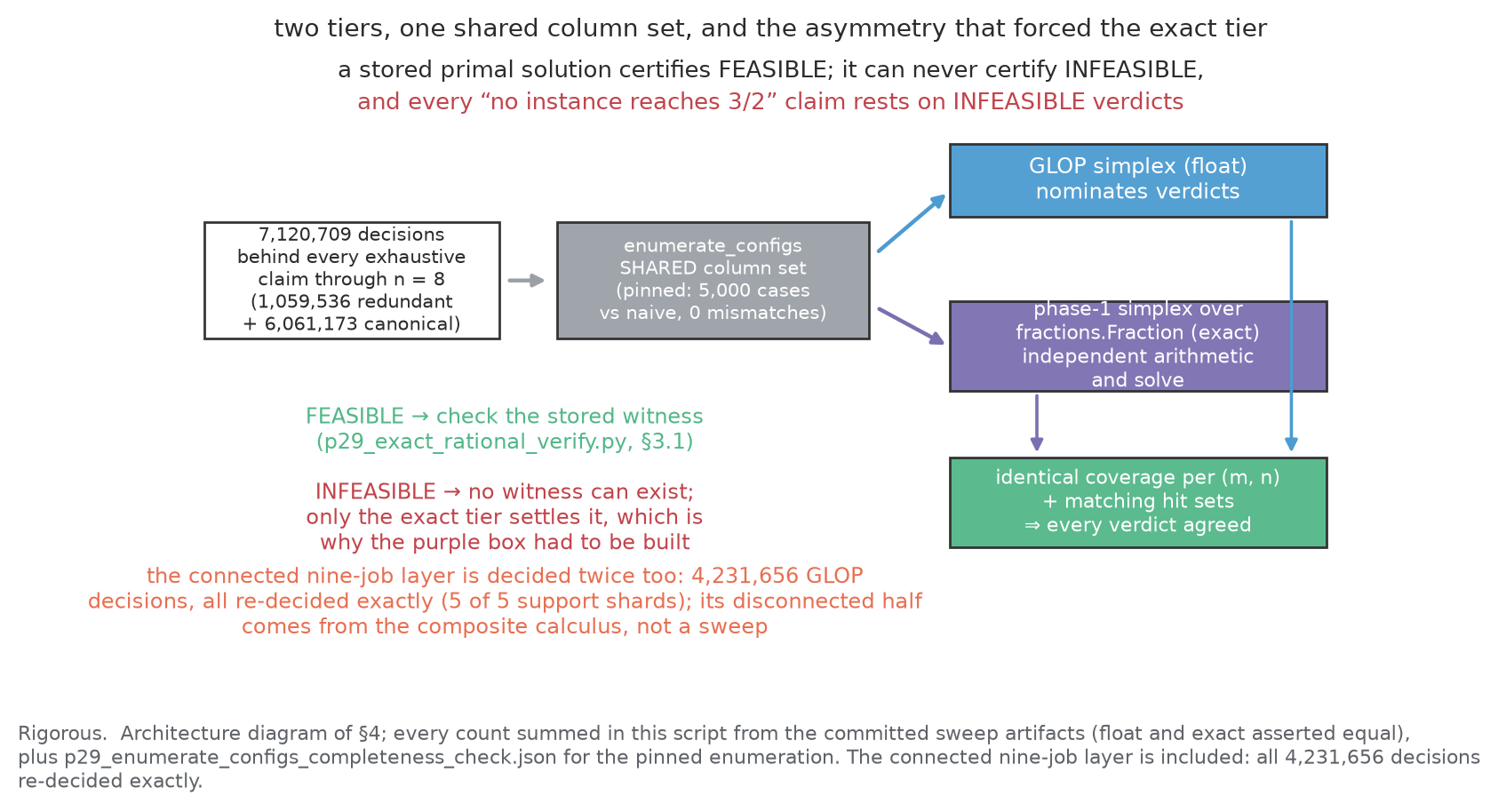}
\caption{The two verification tiers and the feasible/infeasible asymmetry
that forced the exact-rational simplex.}
\label{fig:tiers}
\end{figure}

\begin{figure}[htbp]\centering
\includegraphics[width=0.95\textwidth]{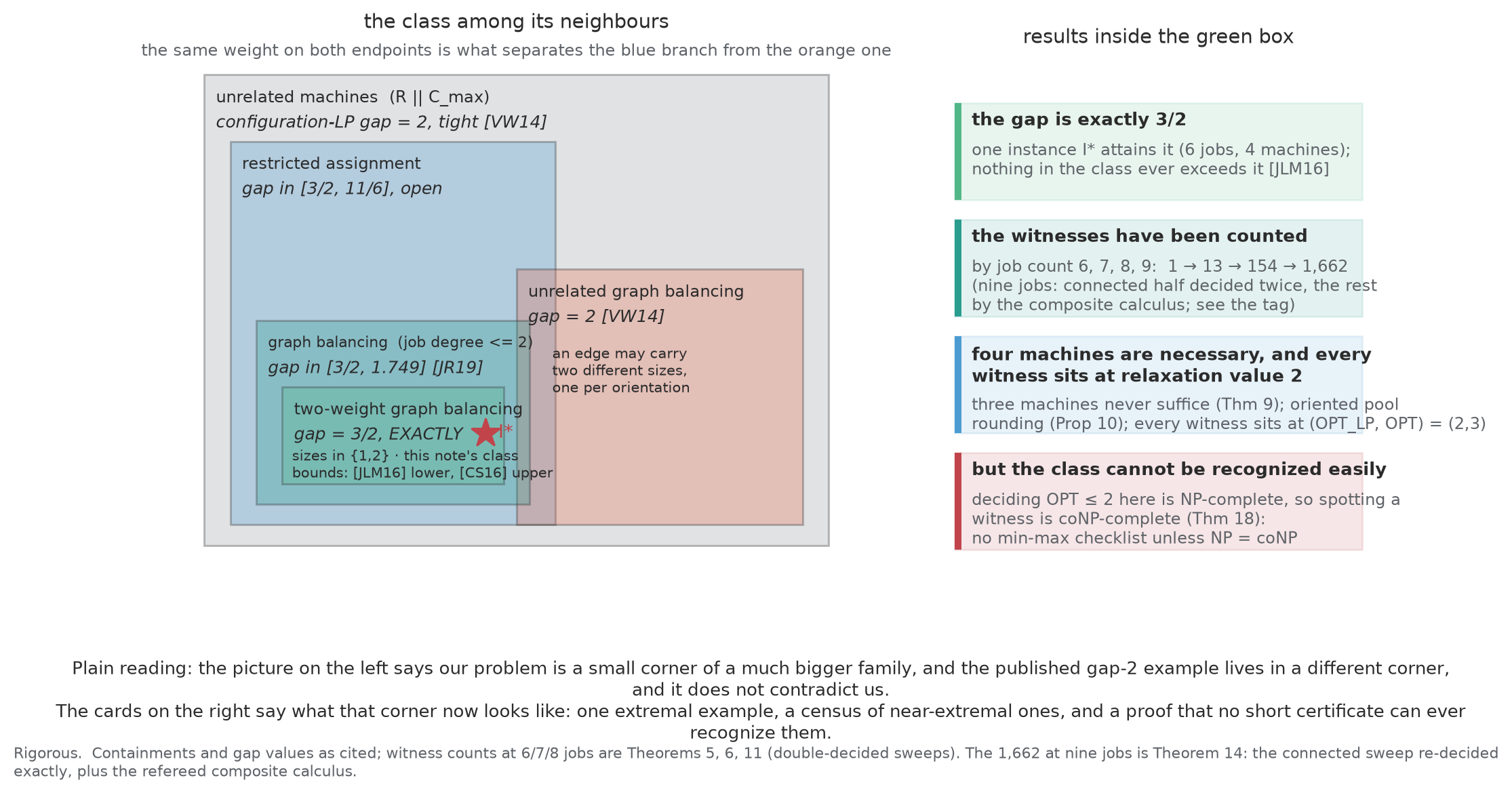}
\caption{Two-weight graph balancing among its neighbouring classes. Verschae--Wiese's gap-2
construction lives in the unrelated model, not in the class studied here.}
\label{fig:landscape}
\end{figure}

\begin{figure}[htbp]\centering
\includegraphics[width=0.9\textwidth]{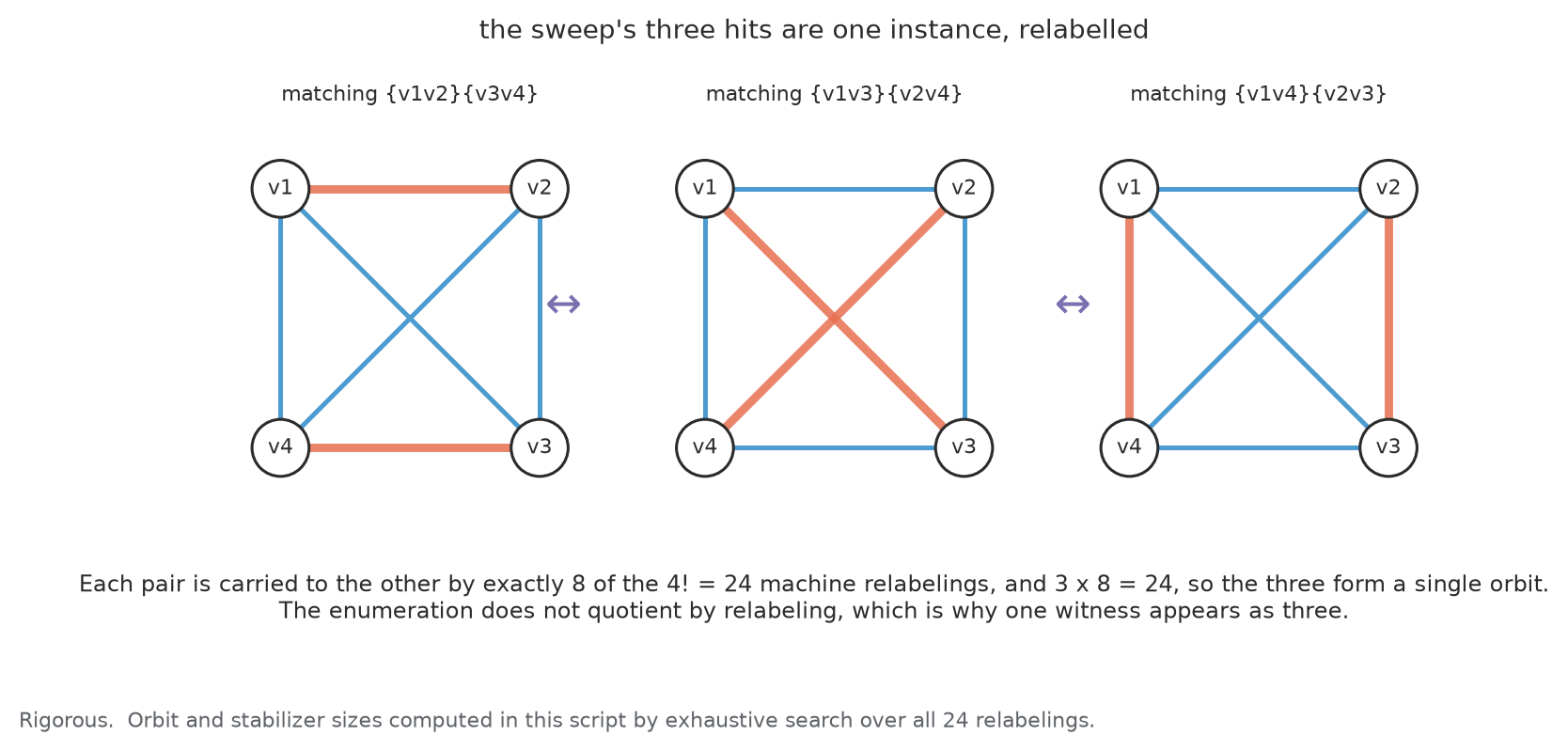}
\caption{Uniqueness as an orbit: the three hits are the three perfect
matchings of $K_4$, one instance under relabeling.}
\label{fig:orbit}
\end{figure}

\begin{figure}[htbp]\centering
\includegraphics[width=0.9\textwidth]{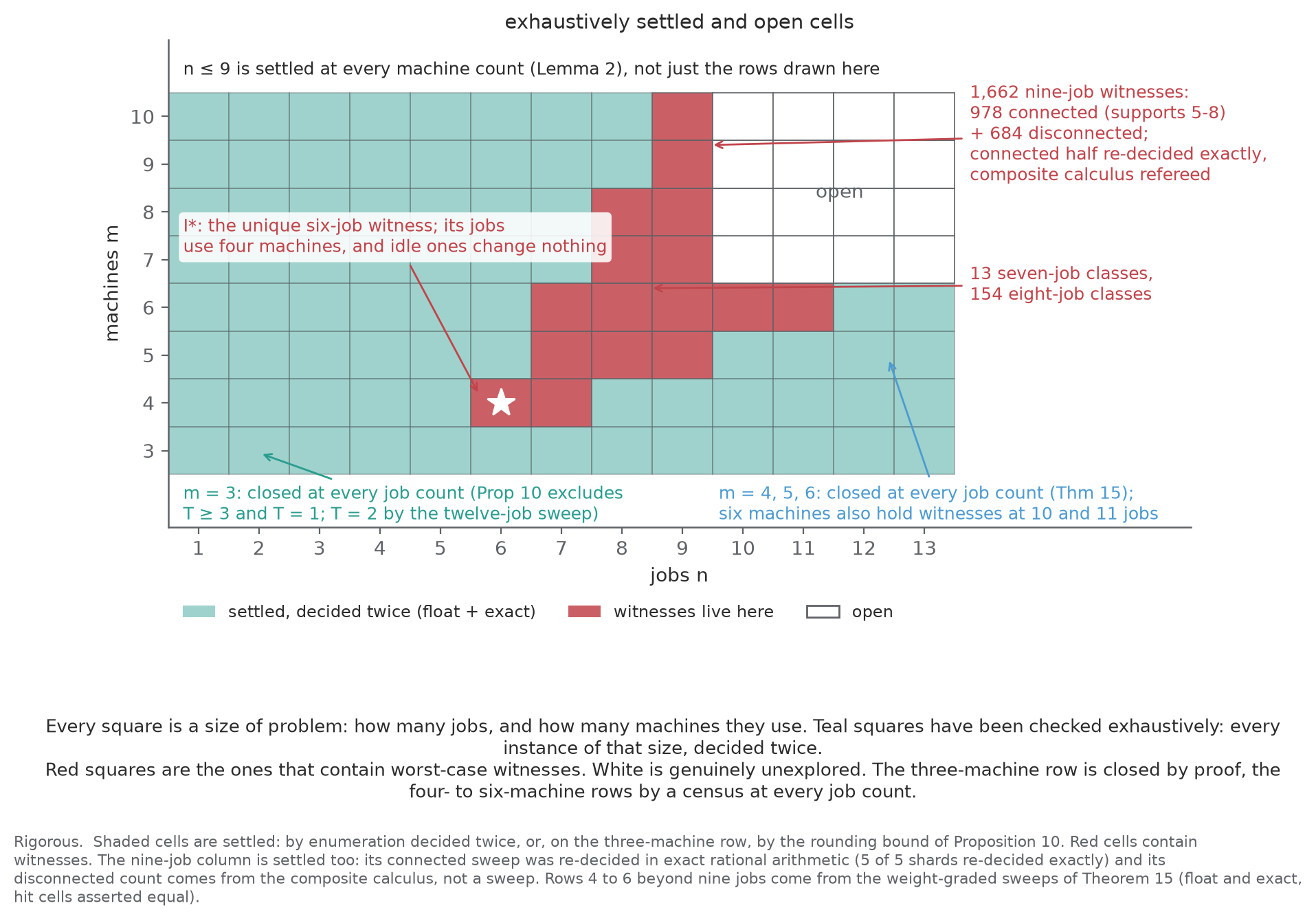}
\caption{Exhaustively settled and open regions of the (machines, jobs) grid,
a row being the number of machines the jobs can use. Every cell up to nine
jobs is settled (Lemma~\ref{lem:conn}); the three-machine row is closed by
Theorem~\ref{prop:mach}, and the four- to six-machine rows at every job count by
Theorem~\ref{thm:machaxis}, whose six-machine sweep also finds witnesses at ten
and eleven jobs ($1{,}459$ and $1{,}027$ classes).}
\label{fig:frontier}
\end{figure}

\begin{figure}[htbp]\centering
\includegraphics[width=0.95\textwidth]{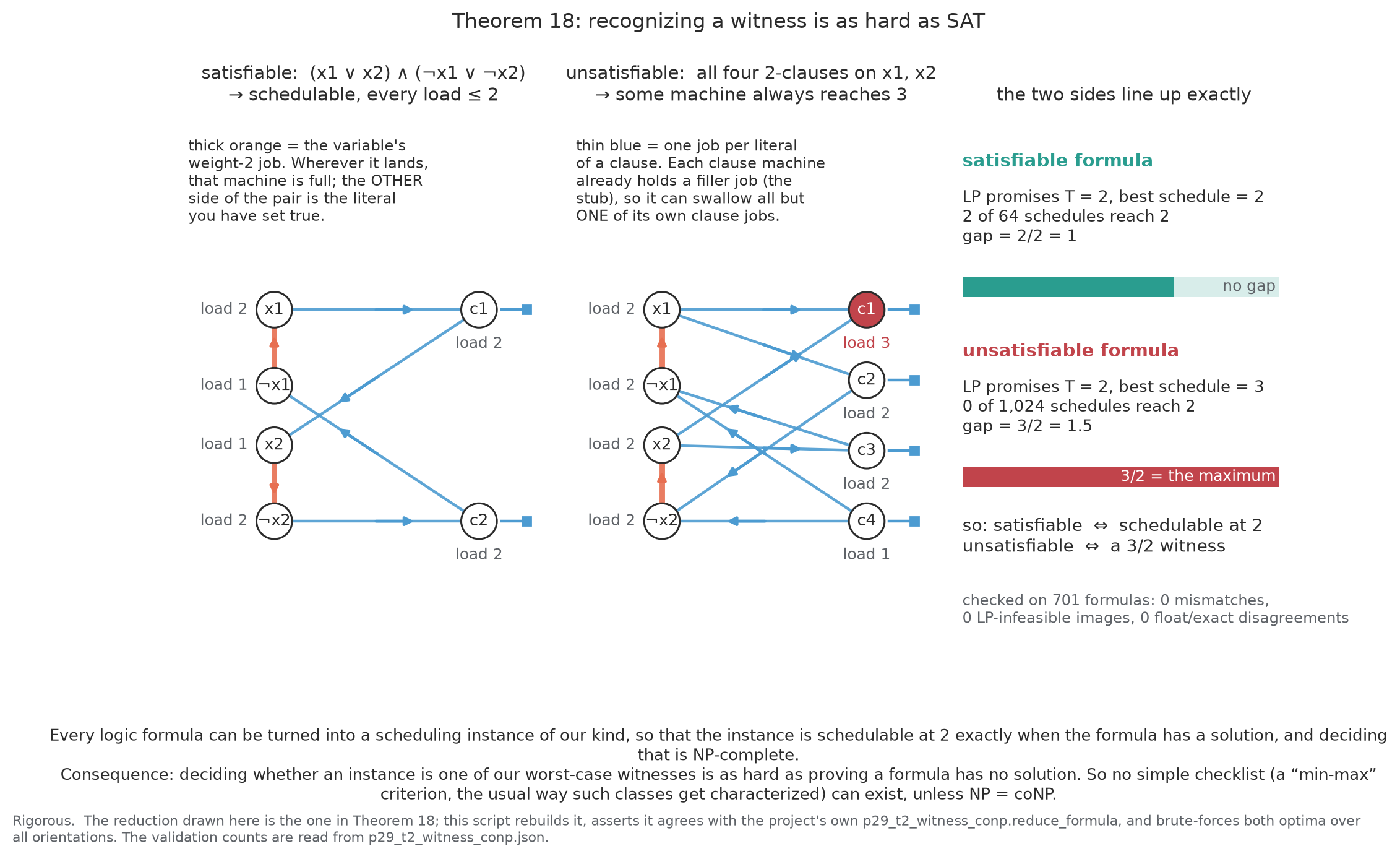}
\caption{Theorem~\ref{thm:conp} drawn. A satisfiable formula's image is
schedulable at 2; the unsatisfiable core forces a clause machine to 3,
making it a witness. Deciding which case one is in is therefore
NP-complete, so no min--max criterion with polynomially checkable
conditions can characterize witnesses unless $\mathrm{NP} = \mathrm{coNP}$.}
\label{fig:conp}
\end{figure}

\begin{figure}[htbp]\centering
\includegraphics[width=0.95\textwidth]{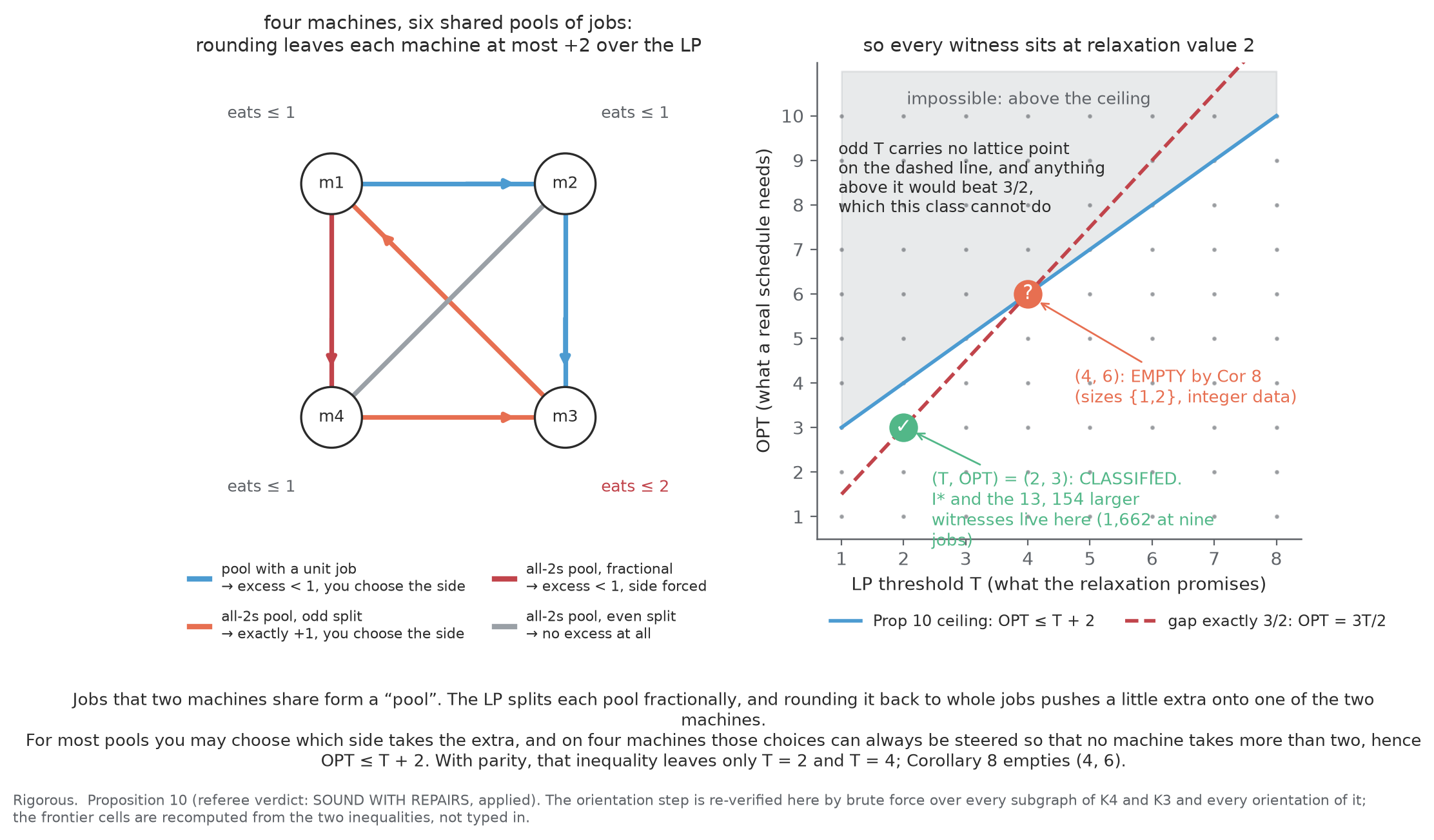}
\caption{Proposition~\ref{prop:orient}. Left: the six pair pools of four
machines, with the excess each contributes and whether its direction is
choosable; steering the choosable ones keeps every machine at $+2$.
Right: the inequality $\OPT \le T+2$ against gap $3/2$. It leaves
$(T,\OPT) = (4,6)$ standing on its own; Corollary~\ref{cor:cell} empties
that cell.}
\label{fig:pools}
\end{figure}

\begin{figure}[htbp]\centering
\includegraphics[width=0.97\textwidth]{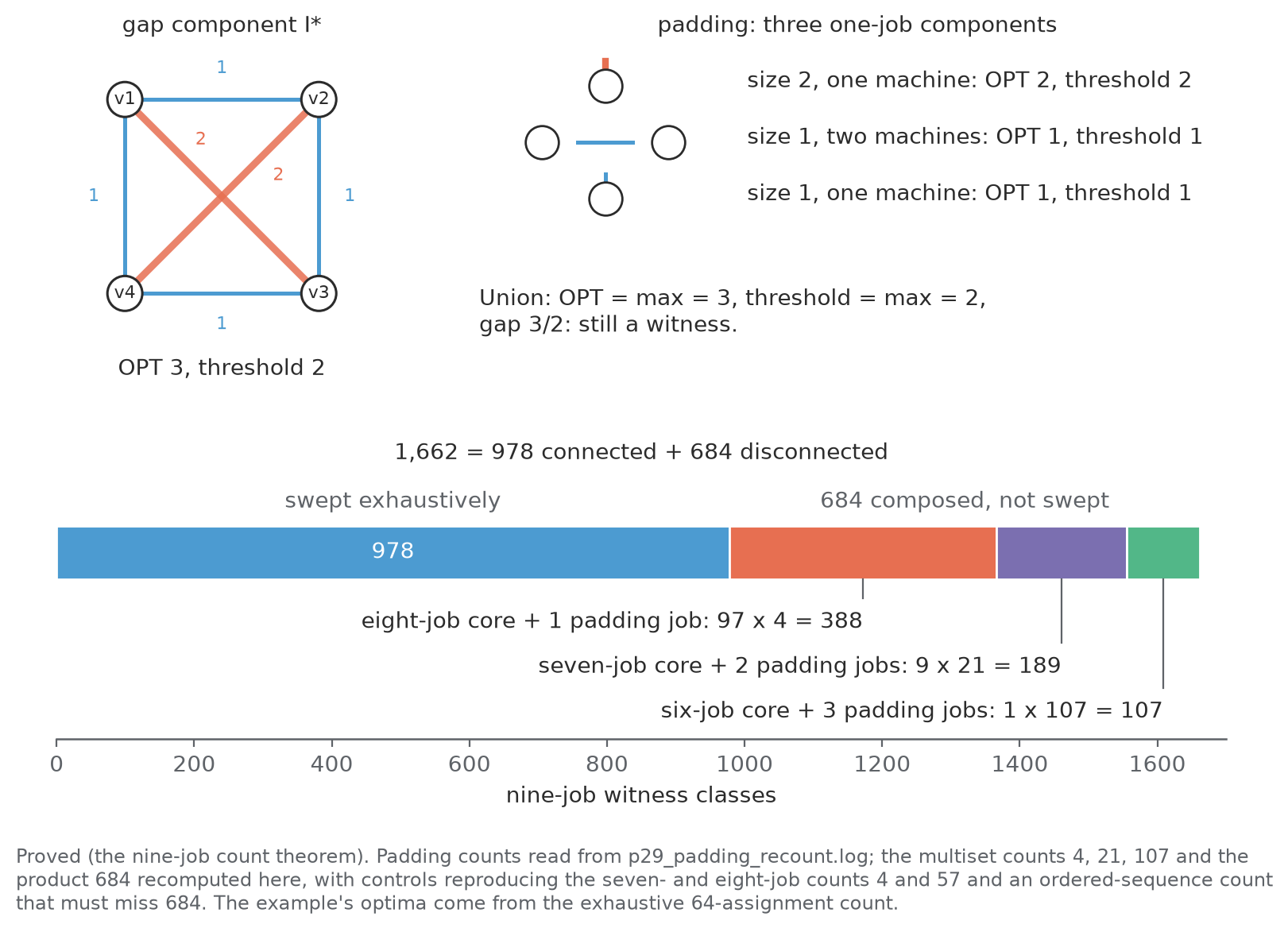}
\caption{The composite calculus of Theorem~\ref{thm:nine}. \textbf{Left:} the gap
component $I^*$ with three one-job padding components beside it. A disconnected
instance's optimum and relaxation value are the largest over its components
(Lemma~\ref{lem:decomp}), so the union keeps $\OPT = 3$ and $\OPTLP = 2$ and is
still a witness. \textbf{Right:} the $1{,}662$ nine-job witness classes. The
$978$ connected classes come from the exhaustive sweep; the $684$ disconnected
ones are composed, one gap component of eight, seven or six jobs times the number
of padding multisets carrying the remaining one, two or three jobs:
$97 \times 4 + 9 \times 21 + 1 \times 107$. The padding counts are read from the
deposited recount, and the products are recomputed when the figure is drawn.}
\label{fig:calculus}
\end{figure}

\begin{figure}[htbp]\centering
\includegraphics[width=0.95\textwidth]{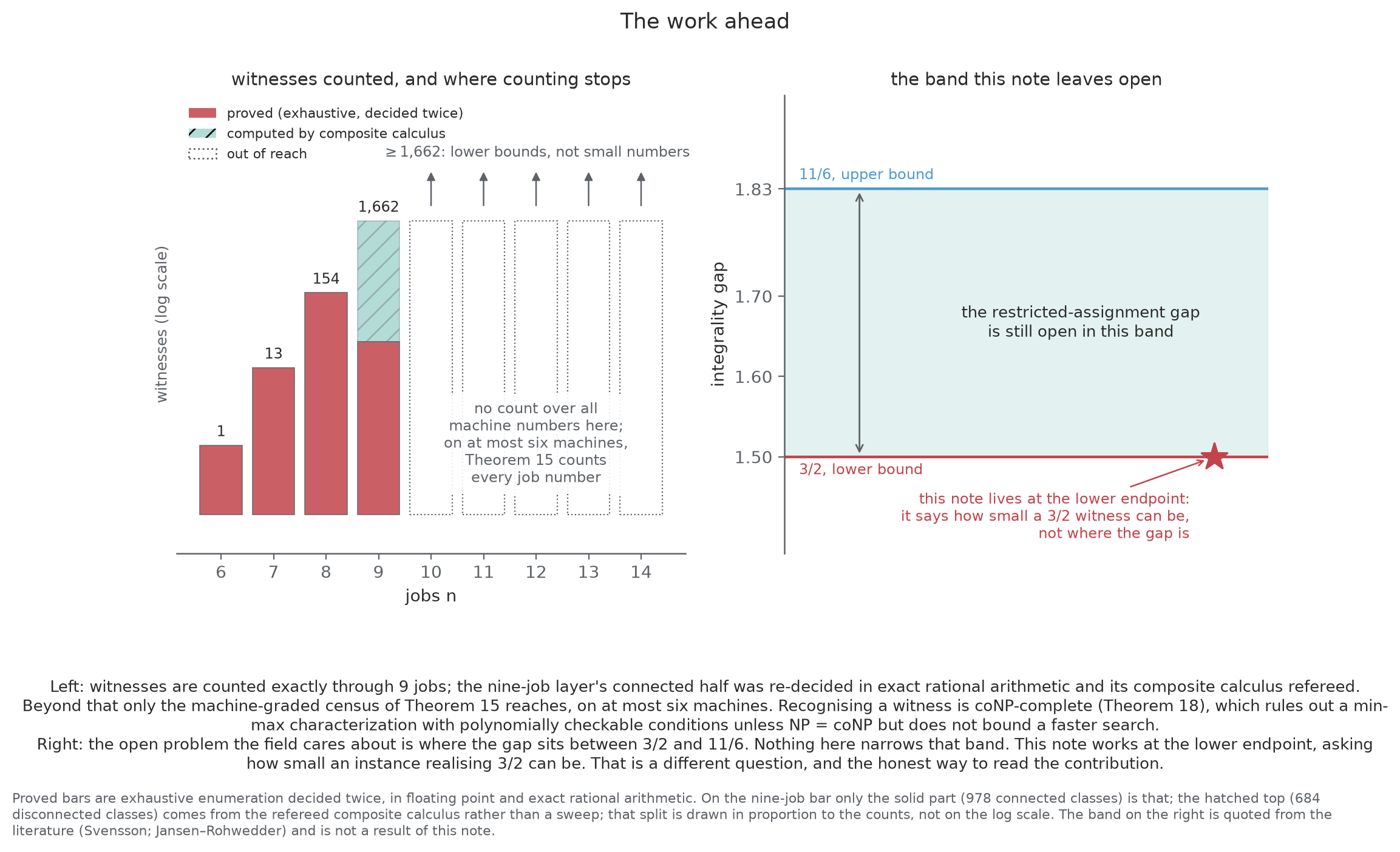}
\caption{The work ahead. Left: witnesses are counted exactly
through nine jobs---the nine-job layer's connected half re-decided in
exact rational arithmetic, its disconnected half by the composite calculus
of Theorem~\ref{thm:nine}, checked separately during adversarial review,
rather than swept---and
beyond it only the machine-graded census of Theorem~\ref{thm:machaxis} reaches,
on at most six machines. Theorem~\ref{thm:conp} rules out a
min--max characterization of witnesses with polynomially checkable conditions
unless $\mathrm{NP} = \mathrm{coNP}$; it
does not bound how far a faster search could go. Right: the open
problem is where the restricted-assignment gap actually lies in
$[3/2, 11/6]$, drawn vertically so the smaller gap sits lower. Nothing
here narrows that band. This note works at the \emph{lower}
endpoint, asking how small an instance realizing $3/2$ can be, which is a
different question.}
\label{fig:future}
\end{figure}

\clearpage
\section*{Acknowledgements}

I thank Baruch Schieber (NJIT) for explaining
the restricted-assignment configuration-linear-program gap problem to me and
for challenging me to tackle it. I thank JMS and CS for reading and
commenting on drafts.

\section*{Use of AI tools}

AI tools assisted under the author's direction. Claude (Anthropic) wrote search,
verification and figure code to the author's specification. Claude and
ChatGPT (OpenAI) gave editorial support. Perplexity was used for literature search and citation verification.
ChatGPT (OpenAI) and Gemini (Google) were used for adversarial review and
independent mathematical checking, and Wolfram Mathematica for checking
derivations.

\section*{Author's note}

I offer this work in the spirit of community-supported research---the
distributed protein-folding projects are the model I have in mind---where
the computing is contributed rather than bought, the claim-bearing record is
public, and the worth of the result is what others can build on it. The
public replication package maps every computational result reported here,
including searches that found nothing, to the code that produces it and to
a stored artifact or log where the run creates one. Additional developmental
materials, including superseded approaches and unreported exploratory
searches, are available upon reasonable request, subject to privacy and
copyright constraints. The negative results are most of the evidence here.
I hope the work is of use to the field.

\vfill
\begin{center}
\begin{minipage}{0.72\textwidth}
\itshape\small
Four swift lathes shape a silent silver square,\\
Spinning four light threads around two heavy cables.\\
The continuous LP logic seals at two,\\
While true integer schedules rise to three.
\end{minipage}
\end{center}

\end{document}